\pdfoutput=1
\documentclass[11pt]{article}
\usepackage{bbm}
\usepackage{amsmath}
\usepackage{amsthm}
\usepackage{amssymb} 
\usepackage{mathrsfs}
\usepackage[dvipsnames]{xcolor}
\usepackage{graphicx}
\usepackage{authblk}
\usepackage[backref=none,hypertexnames=false,colorlinks=true,urlcolor=blue,linkcolor=blue,citecolor=blue]{hyperref}
\usepackage[numbers,comma,square,sort&compress]{natbib}
\usepackage[letterpaper,text={6.5in,9in},centering]{geometry}
\usepackage{tabularx}
\usepackage{booktabs} 
\usepackage{array} % for tables

\graphicspath{{eps/}{pdf/}}

\makeatletter\@addtoreset{equation}{section}\makeatother

\newtheorem{thm}{Theorem}[section] 
\newtheorem{lem}[thm]{Lemma}  
\newtheorem{cor}[thm]{Corollary} 
\newtheorem{prop}[thm]{Proposition}  
\newtheorem{hyp}{Hypothesis}
\newtheorem{defn}[thm]{Definition}

\newcommand{\rev}[1]{{\color{black}#1}}

\begin{document}

\title{
Stability of localized solutions to lattice dynamical systems
}
\author{Bocheng Ruan}
\author{Jack M. Hughes}
\author{Jason J. Bramburger}

\affil{\small Department of Mathematics and Statistics, Concordia University, Montr\'eal, QC, Canada}

\date{}
\maketitle

\begin{abstract}
Localized patterns are spatially confined structures that arise in lattice dynamical systems and play an important role in physics, biology, and materials science. While their existence and bifurcation structure are well-understood, the stability of these solutions remains largely unexplored, particularly in discrete and high-dimensional settings. In this work, we develop a general theoretical framework to analyze the spectral stability of localized steady states in one-dimensional and multi-dimensional rectangular lattices. Our approach leverages the properties of front and back solutions, combined with a discrete Evans function formulation, to characterize the spectrum of localized solutions. \rev{We prove that, for well-separated regions of localization, the point spectrum of the underlying fronts and backs are asymptotically inherited by the localized patterns. Furthermore, we demonstrate that for commonly studied Nagumo-type lattices, this completely characterizes the stability.} This framework applies to solutions with single or multiple plateaus, including oscillatory and multi-pulse configurations. We illustrate the results on a real-valued cubic–quintic Ginzburg–Landau lattice, a prototypical Nagumo-type system, and provide numerical demonstrations of bifurcation structures and eigenvalue spectra.
\end{abstract}

%\keywords{}

%%%%%%%%%%%%%%%%%%%%%%%%%%%%%%%%%%%%%%%%%
\section{Introduction} \label{sec:intro}

Localized patterns play a central role in a wide range of spatially extended systems arising in applications across physics, biology, and materials science. Their formation and structure have been extensively studied over the past several decades, resulting in a well-developed theory describing the existence and bifurcations of such states in both continuous and discrete settings. In contrast, their stability properties remain comparatively less understood. We refer the reader to the recent review \cite{bramburger2024localized}, as well as earlier surveys \cite{dawes2010emergence,knobloch2015spatial}, for comprehensive overviews of localized patterns, their bifurcation structure, and applications.

\begin{figure}[t] %Figure: Steady-states of the LDS
    \center
    \includegraphics[width = 0.99\textwidth]{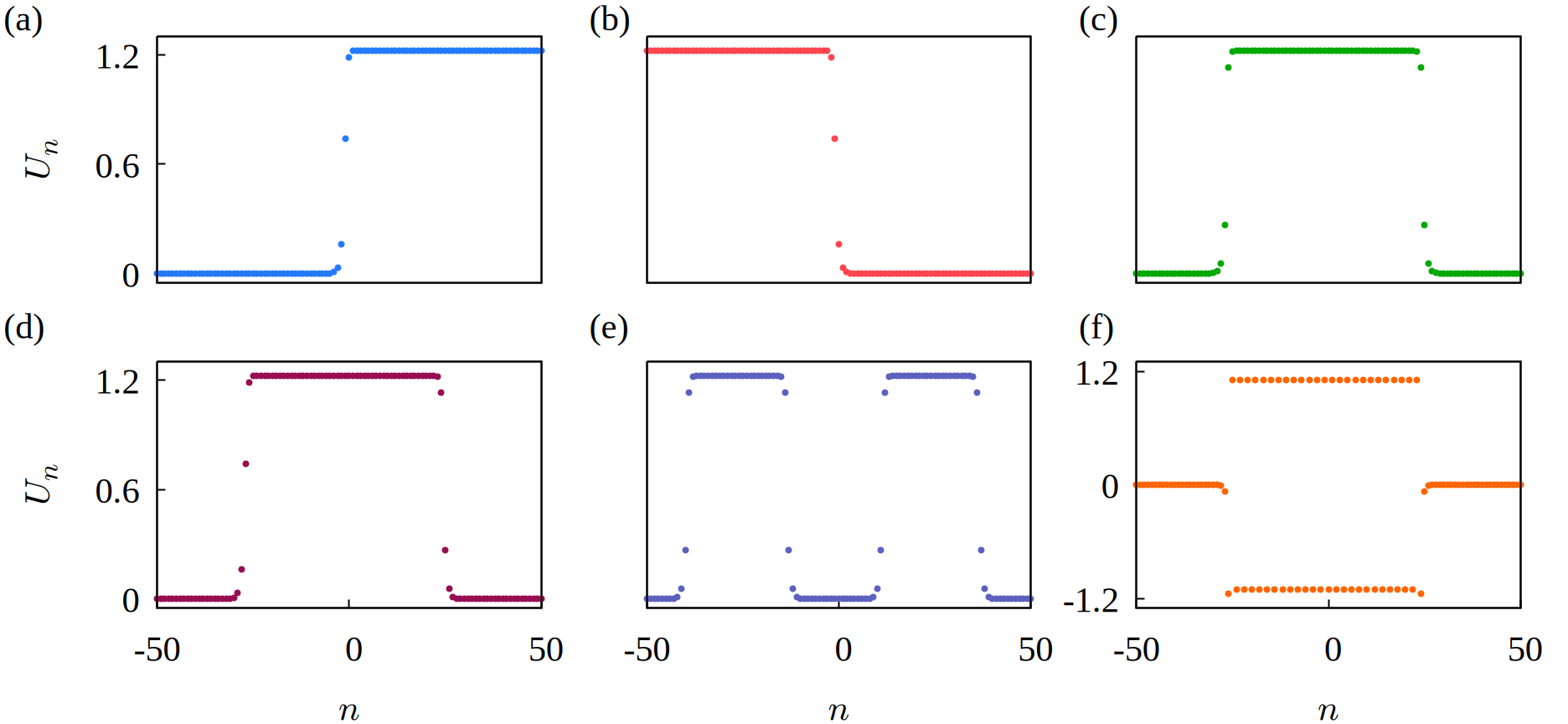}  
    \caption{Steady-state solutions to \eqref{LDS} with $f(u) = -0.75u + 2u^3 - u^5$, $\theta = 0.5$ in (a-e) and $\theta=0.05$ in (f). (a) Fronts and (b) backs can be combined to construct (c) symmetric and (d) asymmetric localized solutions. Other possibilities include (e) multiple regions of localization (here a 2-pulse solution) and (f) oscillatory regions of localization. }
    \label{fig:IntroFig}
\end{figure}

In this paper, we study localized steady-state solutions in lattice dynamical systems, focusing on prototypical Nagumo-type lattices \cite{chow2003lattice}
\begin{equation}\label{LDS}
    \dot U_n = \theta(U_{n+1} + U_{n-1} - 2U_n) + f(U_n,\mu), \qquad n\in\mathbb{Z},
\end{equation}
where $U_n = U_n(t)$, $\theta > 0$ is the coupling between neighboring lattice sites, and $f:\mathbb{R}\times\mathbb{R}\to\mathbb{R}$ is bistable for an open, bounded range of $\mu$. We focus on steady-states $\dot U_n = 0$, which can be recast as a discrete dynamical system in $\mathbb{R}^2$ by setting $(u^{(1)}_n,u^{(2)}_n) = (U_{n-1},U_n)$:
\begin{equation}\label{LDSmap}
    \begin{split}
        u^{(1)}_{n+1} &= u^{(2)}_n, \\
        u^{(2)}_{n+1} &= -u^{(1)}_n + 2u^{(2)}_n - \frac{1}{\theta}f(u^{(2)}_n,\mu).
    \end{split}
\end{equation}
Bounded orbits of \eqref{LDSmap} correspond precisely to steady-states of \eqref{LDS}.

One of the most striking features of lattice systems is propagation failure \cite{keener1987propagation,guo2011wave,chow1996dynamics,hupkes2011propagation}. In contrast to PDEs, lattice discreteness can halt traveling waves, giving rise to stationary front and back solutions for sufficiently small coupling $\theta$. In the map \eqref{LDSmap}, these stationary solutions correspond to heteroclinic orbits (Figure~\ref{fig:IntroFig}(a)–(b)). Recent results \cite{bramburger2020spatially,bramburger2021isolas} show that each front and back pair implies the existence of infinitely many localized steady states at the same parameter values. These localized solutions correspond to homoclinic orbits of \eqref{LDSmap} that connect the trivial fixed point (Figure~\ref{fig:IntroFig}(c)–(f)).

The central question of this work is the stability of these localized solutions. Linearizing about a steady state $\{U_n^*\}_{n\in\mathbb{Z}}$ with a perturbation $U_n = U_n^* + V_n e^{\lambda t}$ leads to the eigenvalue problem
\begin{equation}\label{LDSlinear}
    \lambda V_n = \theta(V_{n+1} + V_{n-1} - 2V_n) + f'(U_n^*)V_n, \qquad n\in\mathbb{Z},
\end{equation}
which can be rewritten as the non-autonomous linear map
\begin{equation}\label{LDSeig}
        \begin{bmatrix}
            v^{(1)}_{n+1} \\ v^{(2)}_{n+1}     
        \end{bmatrix} = \begin{bmatrix}
            0 & 1 \\ -1 & 2 - \frac{(f'(U_n^*) - \lambda)}{\theta}
        \end{bmatrix}\begin{bmatrix}
            v^{(1)}_n \\ v^{(2)}_n     
        \end{bmatrix}.
\end{equation}
For $\lambda \in \mathbb{C}$, the existence of a nontrivial bounded solution of \eqref{LDSeig} is equivalent to $\lambda$ belonging to the (point) spectrum of the linearization \eqref{LDSlinear}. This correspondence is made precise in Section~\ref{sec:Results} that follows.

Our method exploits the structure of front and back solutions to infer the stability of the localized states they generate. \rev{This approach has important precedents in the continuous setting, including the connection between eigenvalue counts, Chern numbers, and Evans functions \cite{alexander1990waves}, and its use in the stability analysis of composite pulses \cite{alexander1993pulse3,alexander1994pulse2}. Later applications include optical-fibre multipulses \cite{sandstede1997optical} and FitzHugh--Nagumo multifronts \cite{sandstede1998FHN}.} The key analytical tool here is the discrete Evans function \cite{balmforth2000being,kapitula2001stability}, whose zeros coincide with eigenvalues of the linearization. For sufficiently large interface separations, this function asymptotically factorizes into contributions from the constituent fronts and backs.

The lattice setting introduces several distinctions and opportunities. First, it addresses lattice dynamical systems, requiring a discrete Evans function formulation. Second, in the small-coupling regime, the structural hypotheses on fronts, backs, and localized solutions can be verified analytically, rather than relying on numerical checks. Third, it handles localized states with \rev{oscillatory plateaus} and multiple disconnected regions of localization (Figure~\ref{fig:IntroFig}(e)–(f)), \rev{the latter of which} naturally produces a corresponding multiplicity of eigenvalues. \rev{Fourth, unlike translation-invariant PDEs, stationary lattice patterns have no symmetry-enforced zero eigenvalue. Thus, stability can be established without a separate near-zero interaction analysis when the constituent fronts and backs are strictly spectrally stable \cite{makrides2019existence,sandstede1997optical,sandstede1998FHN}. A similar distinction underlies spectral inheritance results for multifronts in spatially periodic systems \cite{bengel2026multiple}.} Our results complement earlier studies: unlike case-by-case stability analyses \cite{bramburger2020localized} or computer-assisted proofs \cite{cadiot2025stability}, our approach provides a general, structural understanding of stability across families of solutions.

The remainder of the paper is organized as follows. Section~\ref{sec:Results} introduces the hypotheses and presents Theorems~\ref{thm:SinglePulse} and \ref{thm:2Pulse}, with proofs in Section~\ref{sec:Proofs}. Section~\ref{sec:Application} applies these results to a real-valued Ginzburg--Landau lattice in one dimension and on rectangular lattices. Section~\ref{sec:discussion} concludes with a discussion and directions for future work.

%%%%%%%%%%%%%%%%%%%%%%%%%%%%%%%%%%%%%%%%%%%%%%%%%%%%%%%%%%%%%%%%%%%%%%%%%%%%%%%%%%%%%%%%%%%%%%%%%%%%
\section{Stability of localized solutions}\label{sec:Results}

In this section, we present our main theoretical contributions. Section~\ref{sec:Exponential} introduces exponential dichotomies and their role in characterizing the spectrum of lattice systems linearized about steady-states. In Section~\ref{sec:StabilityResult}, we outline a sequence of hypotheses and their implications, culminating in our main stability result, Theorem~\ref{thm:SinglePulse}, for localized structures with a single connected region of localization. Section~\ref{sec:MultiResult} extends this result to steady states with two disconnected regions of localization and briefly discusses the extension to multiple regions, which follows by routine, though increasingly tedious, arguments of the same type.

%%%%%%%%%%%%%%%%%%%%%%%%%%%%%%%%%%%%%%%%%%%%%%%%%%%%%%%%%%%%%%%%%%%%%%%%%%%%%%%%%%%%%%%%%%%%%%%%%%%%
\subsection{Exponential dichotomies and the spectrum}\label{sec:Exponential}

Throughout this subsection we consider a linear non-autonomous difference equation 
\begin{equation}\label{LinearAn}
    v_{n+1} = A_nv_n, \quad n\in\mathbb{Z},
\end{equation}
where each $A_n \in \mathbb{R}^{d\times d}$ is nonsingular. The solution operator is denoted
\begin{equation}
    \Phi(n,m) = \begin{cases}
        A_{n-1}\cdots A_m & n > m \\
        I & n = m \\
        A_n^{-1}\cdots A_{m-1}^{-1} & n < m.
    \end{cases}
\end{equation}
Let us denote $\mathbb{Z}_-$ as the nonpositive integers and $\mathbb{Z}_+$ as the nonnegative integers. We present the following definition from \cite{beyn1997numerical}.

\begin{defn}\label{def:ExpDich}
    The difference equation \eqref{LinearAn} has an {\bf exponential dichotomy} on $J \in \{\mathbb{Z}_-,\mathbb{Z}_+,\mathbb{Z}\}$ if there exists projectors $P(n)$, $n\in J$, in $\mathbb{R}^d$ and constants $K,\alpha > 0$ such that 
    \[
        P(n)\Phi(n,m) = \Phi(n,m)P(m), \quad \forall n,m \in J,
    \]
    and 
    \[
        \begin{split}
            \|\Phi(n,m)P(m)\| &\leq K\mathrm{e}^{-\alpha (n-m)}, \\
            \|\Phi(m,n)(I - P(n))\| &\leq K\mathrm{e}^{-\alpha (n-m)},
        \end{split}
    \]
    for all $n\geq m$ in $J$.
\end{defn}

With the definition of an exponential dichotomy now presented, we turn to parameter dependent linear difference equations. Indeed, inspired by \eqref{LDSeig}, consider the difference equation
\begin{equation}\label{LinearAnLambda}
    v_{n+1} = A_n(\lambda)v_n, \quad n\in\mathbb{Z},
\end{equation}
for which the nonsingular $A_n(\lambda) \in\mathbb{R}^{d \times d}$ depend continuously on a parameter $\lambda \in \mathbb{C}$. Considering the Hilbert space 
\[
    \ell^2 = \{\{v_n\}_{n\in\mathbb{Z}}|\ \sum_{n\in\mathbb{Z}} |v_n|^2 < \infty\},
\]
we also define the linear mapping $\mathcal{T}(\lambda)$ acting on $v = \{v_n\}_{n\in\mathbb{Z}} \in \ell^2$ by
\begin{equation}
    [\mathcal{T}(\lambda)v]_n = v_{n+1} - A_n(\lambda)v_n, \quad \forall n \in \mathbb{Z}.
\end{equation}
We say that $\lambda \in\mathbb{C}$ is in the spectrum of $\mathcal{T}$ if $\mathcal{T}(\lambda)$ is not invertible on $\ell^2$ for the given value of $\lambda$. 

The spectrum of $\mathcal{T}$ is further decomposed into two distinct components: the {\bf point spectrum} and the {\bf essential spectrum}. The point spectrum is all such $\lambda \in \mathbb{C}$ for which $\mathcal{T}(\lambda)$ is not invertible with Fredholm index zero, while the essential spectrum is the remaining elements of the spectrum. The connection between $\mathcal{T}(\lambda)$ and \eqref{LinearAnLambda} is that $\lambda$ does not belong to the essential spectrum of $\mathcal{T}$ if and only if \eqref{LinearAnLambda} has an exponential dichotomy on both $\mathbb{Z}_-$ and $\mathbb{Z}_+$ and the associated projectors $P_-(n)$ and $P_+(n)$, respectively, satisfy $\mathrm{rank}(P_-(n)) = \mathrm{rank}(P_+(n))$. The values $i_{\pm}(\lambda) = \mathrm{rank}(P_\pm(n))$ are called Morse indices, while the Fredholm index of $\mathcal{T}(\lambda)$ is given by $\mathrm{rank}(P_+(n)) - \mathrm{rank}(P_-(n))$.

Assuming that $A_n(\lambda) \to A_\pm(\lambda) \in \mathbb{R}^{d\times d}$ as $n \to \pm \infty$, the essential spectrum is characterized by the asymptotic matrices $A_{\pm}(\lambda)$. The two possibilities are as follows:
\begin{enumerate}
    \item At least one of $A_+(\lambda)$ or $A_-(\lambda)$ has eigenvalues in the unit circle at $\lambda \in \mathbb{C}$. Thus, $\mathcal{T}(\lambda)$ is not invertible at this value of $\lambda$, and so $\lambda$ belongs to the essential spectrum of $\mathcal{T}$.
    \item If both $A_{\pm}(\lambda)$ have no eigenvalues on the unit circle at $\lambda \in \mathbb{C}$, then $\mathcal{T}(\lambda)$ is Fredholm. If the Fredholm index is non-zero, then this $\lambda$ belongs to the essential spectrum.  
\end{enumerate}
Thus, we see that the essential spectrum can be characterized by simply examining the asymptotics of the sequence of parametrized matrices $\{A_n(\lambda)\}_{n \in \mathbb{Z}}$. The point spectrum is more difficult to characterize as it corresponds to having an eigenvector $\mathcal{T}(\lambda)v = 0$. Such an eigenvector takes the form of a bounded solution of \eqref{LinearAnLambda} belonging to $\ell^2$.

\paragraph{Application to lattice systems.} Recall that our interest in this work is to characterize the stability of solutions to lattice dynamical systems. Returning to the prototypical system \eqref{LDS}, we saw that the stability of a bounded solution $\{U_n^*\}_{n \in \mathbb{Z}}$ comes from analyzing the eigenvalue equation \eqref{LDSlinear}. Introducing the transformation $(v^{(1)},v^{(2)}) = (V_{n-1},V_n)$ results in the $\lambda$-dependent linear maps \eqref{LDSlinear} having
\begin{equation}
    A_n(\lambda) = \begin{bmatrix}
            0 & 1 \\ -1 & 2 - \frac{(f'(U_n^*) - \lambda)}{\theta}
        \end{bmatrix},
\end{equation}
while the operator $\mathcal{T}(\lambda)$ is equivalent to \eqref{LDSeig}. If $U_n^* \to U^*_\pm$ as $n\to\pm\infty$ then the asymptotic matrices become
\begin{equation}\label{Apm}
    A_\pm(\lambda) = \begin{bmatrix}
            0 & 1 \\ -1 & 2 - \frac{(f'(U^*_\pm) - \lambda)}{\theta}
        \end{bmatrix},
\end{equation}
which can be used to characterize the essential spectrum as follows.

\begin{lem}\label{lem:EssentialSpectrum}
    Suppose $\{U_n^*\}_{n \in \mathbb{Z}}$ is a bi-infinite sequence satisfying $U_n^* \to U^*_\pm$ as $n\to\pm\infty$. Then, the essential spectrum of \eqref{LDSlinear} is given by 
    \[
        \{\lambda \in \mathbb{C}|\ \mathrm{Re}(\lambda) \in [f'(U^*_-) - 4\theta,f'(U_-^*)] \cup [f'(U^*_+) - 4\theta,f'(U_+^*)]\ \&\ \mathrm{Im}(\lambda) = 0\}.
    \] 
\end{lem}

\begin{proof}
    First, we note that the right hand side of \eqref{LDSlinear} is self-adjoint as an operator on $\ell^2$. Thus, its spectrum is purely real-valued, meaning we can restrict the following analysis to $\lambda \in \mathbb{R}$ to characterize its spectrum.
    
    Now, consider the constant asymptotic matrix $A_+(\lambda)$ in \eqref{Apm} with $\lambda \in \mathbb{R}$. The eigenvalues of $A_+(\lambda)$ are
    \begin{equation}
        \mu_{1,2} = \frac{2\theta - (f'(U_+^*) - \lambda) \pm \sqrt{(2\theta - (f'(U_+^*) - \lambda))^2 -4\theta^2}}{2\theta}.
    \end{equation}
    If $(2\theta - (f'(U_+^*) - \lambda))^2 -4\theta^2 > 0$ then the eigenvalues are real with $|\mu_2| < 1 < |\mu_1|$, giving that \eqref{LDSmap} has an exponential dichotomy on $\mathbb{Z}_+$ with Morse index $\mathrm{rank}(P_+(n)) = 1$. Alternatively, if
    \[
        (2\theta - (f'(U_+^*) - \lambda))^2 -4\theta^2 \leq 0 \implies \lambda \in [f'(U_+) - 4\theta,f'(U_+)]
    \]
    then $A_+(\lambda)$ has spectrum on the unit circle and so the interval $[f'(U_+) - 4\theta,f'(U_+)]$ belongs to the essential spectrum of \eqref{LDSlinear}.

    An identical argument shows that $A_-(\lambda)$ has eigenvalues on the unit circle if and only if $\lambda \in [f'(U_-) - 4\theta,f'(U_-)]$, which means that this real interval also belongs to the essential spectrum of \eqref{LDSlinear}. Moreover, when $\lambda \notin [f'(U_-) - 4\theta,f'(U_-)]$ the Morse index is again $\mathrm{rank}(P_-(n)) = 1$, giving that the Fredholm index is always zero when $\lambda \notin [f'(U_-) - 4\theta,f'(U_-)] \cup [f'(U_+) - 4\theta,f'(U_+)]$. Thus we have completely characterized the essential spectrum, completing the proof.  
\end{proof}

\rev{Lemma~\ref{lem:EssentialSpectrum} shows that if the asymptotic fixed points of a front (or back) are strictly spectrally stable, then the essential spectrum lies strictly in the left half-plane and does not affect the stability of the front. The same reasoning applies to solutions that asymptote to $k$-cycles. In this case, the asymptotic coefficient matrices in~\eqref{LDSeig} are $k$-periodic, and the essential spectrum is characterized by the corresponding $2\times2$ $k$-step monodromy matrices
\begin{align*}
\mathcal M_\pm(\lambda)=A_{k-1,\pm}(\lambda)\cdots A_{1,\pm}(\lambda)A_{0,\pm}(\lambda).
\end{align*}
Since each matrix \(A_{j,\pm}(\lambda)\) has determinant one, the monodromy matrices have determinant one. Consequently, whenever they are hyperbolic, each has one eigenvalue inside and one eigenvalue outside the unit circle. Thus, the associated difference operator is Fredholm with index zero, and the essential spectrum is determined by the values of $\lambda$ for which at least one of the monodromy matrices has an eigenvalue on the unit circle. Thus, if the $k$-cycle is strictly spectrally stable, then the essential spectrum induced by it lies strictly in the left-half plane.
}

It is important to note that even with an explicit linear equation such as \eqref{LDSmap} to study, there is no indication of what the point spectrum might be. This leads to the work in the next subsection where we leverage information about the point spectrum of front and back solutions to obtain the point spectrum of localized structures in lattice systems.

%%%%%%%%%%%%%%%%%%%%%%%%%%%%%%%%%%%%%%%%%%%%%%%%%%%%%%%%%%%%%%%%%%%%%%%%%%%%%%%%%%%%%%%%%%%%%%%%%%%%
\subsection{Main stability result}\label{sec:StabilityResult}

To provide our main stability result, we begin by considering a diffeomorphism $F:\mathbb{R}^d \to \mathbb{R}^d$ and the iterative scheme
\begin{equation}\label{Fmap}
    u_{n+1} = F(u_n).
\end{equation}
Furthermore, for a bounded solution $\{u_n\}_{n\in\mathbb{Z}}$ of \eqref{Fmap} and a $\lambda \in \mathbb{C}$ we consider the linear mapping about this solution 
\begin{equation}\label{StabilityMap}
    v_{n+1} = [DF(u_n) + B(u_n,\lambda)]v_n, \quad v_n \in \mathbb{C}^d,
\end{equation}
where $DF$ denotes the Jacobian of $F$ and $B:\mathbb{R}^d\rev{\times\mathbb{C}}\to\mathbb{C}^{d\times d}$ depends analytically on $\lambda \in \mathbb{C}$. We assume that both $DF$ and $B$ smoothly depend on their arguments and that the matrices $DF(u) + B(u,\lambda)$ are invertible for all $\lambda \in \mathbb{C}$ and $u \in \mathbb{R}^d$. 

\rev{Allowing $B$ to depend analytically on $\lambda$, rather than requiring it to be linear in $\lambda$, accommodates spatially periodic plateaus. To see this, let $F$ denote the Nagumo spatial map \eqref{LDSmap}. The one-step transfer matrix in \eqref{LDSeig} is $A(u,\lambda) = DF(u) + \lambda C$, where $C = \mathrm{diag}(0,1/\theta)$. For a plateau of spatial period $k$, sampling every $k$th lattice site selects a single point of the cycle, which is a fixed point of $F^k$. The corresponding transfer matrix advances the eigenvalue equation by the same $k$ sites and is given by
\[
    \begin{split}
    A^{(k)}(u,\lambda) &:= A(F^{k-1}(u),\lambda)\cdots A(u,\lambda)\\
     &= D(F^k)(u)+B_k(u,\lambda).
     \end{split}
\]
Indeed, the constant term of this product is $D(F^k)(u)$ by the chain rule, while its leading term is $\lambda^k C^k$. Since $C^k \neq 0$, $B_k(u,\lambda)$ is a matrix polynomial of degree $k$ in $\lambda$. Thus, the sampled eigenvalue equation again has the form \eqref{StabilityMap}, now with $F^k$ in place of $F$ and $B_k$ in place of $B$. This sampling does not discard spectral information: the intermediate values of an eigenfunction are uniquely recovered using the one-step transfer matrices. Along a bounded profile, these matrices and their inverses are uniformly bounded in the lattice index for each fixed $\lambda$. Consequently, sampling and reconstruction preserve boundedness and square summability of eigenfunctions.}

The next assumption provides the existence of possible asymptotic states for bounded solutions $\{u_n^*\}_{n\in\mathbb{Z}}$ of \eqref{Fmap}. In what follows we say a matrix is hyperbolic if it has no eigenvalues on the unit circle in the complex plane.

\begin{hyp}\label{hyp:FixedPoints}
    There exists a nontrivial $u^*\in\mathbb{R}^d$ so that $u = 0,u^*$ are fixed points of \eqref{Fmap}. 
\end{hyp}

\rev{When the framework is applied to a $k$th iterate map, $u^*$ may be a selected point of a $k$-cycle of the original spatial map.} Evaluating the linear map \eqref{StabilityMap} about the constant sequences $\{0\}_{n\in\mathbb{Z}}$ and $\{u^*\}_{n\in\mathbb{Z}}$ leads to the autonomous difference equations
\begin{equation}\label{0starDiffEqns}
    \begin{split}
        v_{n+1} &= A^0(\lambda)v_n, \quad A^0(\lambda) := [DF(0) + B(0,\lambda)], \\
        v_{n+1} &= A^*(\lambda)v_n, \quad A^*(\lambda) := [DF(u^*) + B(u^*,\lambda)].
    \end{split}
\end{equation}
The $n$-independence of the above difference equations coupled with the assumed invertibility of $A^{0,*}(\lambda)$ provides that neither has any point spectra. Moreover, the essential spectra of \eqref{0starDiffEqns} are simply given by the sets 
\begin{equation}
    \begin{split}
        \Sigma^0 &= \{\lambda \in \mathbb{C}|\ A^0(\lambda)\mathrm{\ is\ not\ hyperbolic}\}, \\
        \Sigma^* &= \{\lambda \in \mathbb{C}|\ A^*(\lambda)\mathrm{\ is\ not\ hyperbolic}\}.
    \end{split}    
\end{equation}
Let us set 
\begin{equation}\label{Omega}
    \Omega := \mathbb{C}\setminus \{\Sigma^0 \cup \Sigma^*\},
\end{equation}
which is open since each of $\Sigma^0$ and $\Sigma^*$ are closed subsets of $\mathbb{C}$. Importantly, both difference equations in \eqref{0starDiffEqns} have exponential dichotomies and Fredholm index zero for all $\lambda \in \Omega$. We now make the assumption that there exists heteroclinic connections between $0$ and $u^*$ in \eqref{Fmap}.

\begin{hyp}\label{hyp:FrontBack}
    There exist bounded trajectories $\{u^f_n\}_{n\in\mathbb{Z}}$ and $\{u^b_n\}_{n\in\mathbb{Z}}$ of \eqref{Fmap} along with constants $C,\alpha > 0$ so that
    \begin{equation}
        \begin{split}
            |u_n^f - u^*| &\leq C \mathrm{e}^{-\alpha n}, \\
            |u_n^b| \leq C &\mathrm{e}^{-\alpha n}, 
        \end{split}
    \end{equation}
    for all $n\geq 0$ and 
    \begin{equation}
        \begin{split}
            |u_n^f| \leq C \mathrm{e}^{\alpha n}, \\
            |u_n^b - u^*|\leq C \mathrm{e}^{\alpha n}, 
        \end{split}
    \end{equation}
    for all $n \leq 0$.
\end{hyp}

The above hypothesis posits the existence of heteroclinic connections from $0$ to $u^*$ and vice-versa under the dynamics of $F$. Due to their connection to lattice systems, we term $\{u^f_n\}_{n\in\mathbb{Z}}$ a {\bf front} solution as it represents a step from $0$ to $u^*$ over the lattice $\mathbb{Z}$, while $\{u^b_n\}_{n\in\mathbb{Z}}$ is a {\bf back} representing a step from $u^*$ to $0$; see Figure~\ref{fig:IntroFig} panels (a) and (b), respectively. The following result follows from the roughness theorem for exponential dichotomies \cite[Proposition~2.5]{beyn1997numerical} (see also \cite{sandstede2000absolute,peterhof1997exponential} for similar results in the continuous setting).

\begin{lem}\label{lem:FrontBackEssential}
    For all $\lambda\in\Omega$ the systems 
    \begin{equation}\label{FrontBackLinear}
        v_{n+1} = [DF(u_n^j) + B(u_n^j,\lambda)]v_n, \quad j = f,b,
    \end{equation}
    have exponential dichotomies on $\mathbb{Z}_-$ and $\mathbb{Z}_+$ with projections $P^{-,s}_{j}(n;\lambda)$ and $P^{+,s}_{j}(n;\lambda)$. We denote the associated unstable projections by $P_{j}^{-,u}(n;\lambda) = I - P^{-,s}_{j}(n;\lambda)$ and $P_{j}^{+,u}(n;\lambda) = I - P^{+,s}_{j}(n;\lambda)$.
\end{lem}

Although $\Sigma^0\cup \Sigma^*$ is part of the essential spectrum of linear systems \eqref{FrontBackLinear}, the following hypothesis gives that only the point spectrum can be found in $\Omega$.

\begin{hyp}\label{hyp:Morse}
    For all $\lambda \in \Omega$ the Morse indices are such that
    \begin{equation}
        \mathrm{rank}(P^{-,s}_{j}(n;\lambda)) = \mathrm{rank}(P^{+,s}_{j}(n;\lambda)) = i_\infty,
    \end{equation}
    which in turn gives that 
    \begin{equation}
        \mathrm{rank}(P^{-,u}_{j}(n;\lambda)) = \mathrm{rank}(P^{+,u}_{j}(n;\lambda)) = d - i_\infty,
    \end{equation}
    as well.
\end{hyp}

Hypothesis~\ref{hyp:Morse} indicates that Lemma~\ref{lem:FrontBackEssential} completely characterizes the essential spectrum of the front and back solutions. In practice, it may be the case that Hypothesis~\ref{hyp:Morse} only holds for a subset $\Omega' \subsetneq \Omega$, in which case we can replace $\Omega$ with $\Omega'$ in what follows. \rev{For the scalar Nagumo lattice \eqref{LDS}, Hypothesis~\ref{hyp:Morse} holds with $i_\infty = 1$. Indeed, the one-step transfer matrices in \eqref{LDSeig} have determinant one, as do their $k$-step products. Whenever an asymptotic transfer matrix is hyperbolic, its two eigenvalues therefore have reciprocal moduli, with one inside and one outside the unit circle of the complex plane. Thus, the required stable dimensions agree at both ends, for constant as well as spatially periodic asymptotic states. }

\rev{The set $\Omega$ is a domain of consistent splitting: for every $\lambda\in\Omega$, the asymptotic systems are hyperbolic and the dimensions of the relevant stable and unstable subspaces are constant. 
%Our Evans functions will be defined on local simply connected subsets of $\Omega$, rather than by asserting the existence of a single globally normalized function on all of $\Omega$.
}

\rev{Fix a simply connected open set $U$ with $\lambda_\ast\in U$ and the closure of $U$ in $\Omega$, $\overline U\subset\Omega$. The dichotomy projections can be chosen to depend analytically on $\lambda\in U$. We may therefore choose analytic matrix-valued frames
\[
 Z_f^{-,u}(\lambda),\qquad Z_f^{+,s}(\lambda),
\]
on $U$ whose columns span $\operatorname{Rg}P_f^{-,u}(0;\lambda)$ and $\operatorname{Rg}P_f^{+,s}(0;\lambda)$, respectively. We then define analytic functions on $U$
\begin{equation} \label{EvansFrontBack}
    \begin{split}
        D_f^U(\lambda):=\det\!\left(Z_f^{-,u}(\lambda)\ \ Z_f^{+,s}(\lambda)\right),\\
        D_b^U(\lambda):=\det\!\left(Z_b^{-,u}(\lambda)\ \ Z_b^{+,s}(\lambda)\right).
    \end{split}
\end{equation}
Their normalization is not unique: replacing either analytic frame by another analytic frame multiplies the corresponding function by a nonvanishing analytic factor. Consequently, the locations and algebraic multiplicities of their zeros are independent of the chosen frames. }

The functions $D_f^U$ and $D_b^U$ identify overlap in the range of the unstable projection on $\mathbb{Z}_-$ and the stable projection on $\mathbb{Z}_+$. If $D_j^U(\lambda) = 0$ for some $\lambda$, then there exists an initial condition that converges to $0$ as $n\to \rev{\pm\infty}$ in \eqref{FrontBackLinear}, thus giving an eigenfunction associated to the eigenvalue $\lambda$. Specifically, $\lambda$ is an eigenvalue of multiplicity $m$ for the front (resp. back) solution if and only if $\lambda$ is a root of multiplicity $m$ of $D_f^U$ (resp. $D_b^U$).  
 
\begin{hyp}\label{hyp:Localized}
    There exists a family of bounded trajectories $\{u_n^{\ell}(N)\}_{n \in \mathbb{Z}}$ to \eqref{Fmap}, parametrized by $N \geq N_*$ for some $N_* \geq 1$, satisfying 
    \begin{equation}
        u^\ell_n(N) = \begin{cases}
            u_{n+ N}^f + w^{f,-}_{n+N}(N) & n \leq -N \\
            u_{n+N}^f + w^{f,+}_{n+N}(N) & -N < n \leq 0 \\
            u_{n-N}^b + w^{b,-}_{n-N}(N) & 0 < n \leq N \\
            u_{n-N}^b + w^{b,+}_{n-N}(N) & N < n 
        \end{cases}
    \end{equation}
    and moreover there exist $C,\alpha > 0$  so that
    \begin{equation}
        \begin{split}
            &|w^{f,-}_{n}(N)| \leq C\mathrm{e}^{-\alpha(N-n)}, \qquad n\leq 0 \\
            &|w^{f,+}_{n}(N)| \leq C\mathrm{e}^{-\alpha N}, \qquad 0< n \leq N \\
            &|w^{b,-}_{n}(N)| \leq C\mathrm{e}^{-\alpha N}, \qquad -N< n \leq 0 \\
            &|w^{b,+}_{n}(N)| \leq C\mathrm{e}^{-\alpha (N+n)}, \qquad n> 0. 
        \end{split}
    \end{equation}
\end{hyp}

The solutions $\{u_n^{\ell}(N)\}_{n \in \mathbb{Z}}$ in Hypothesis~\ref{hyp:Localized} represent {\bf localized solutions} since they decay to zero as $n \to \pm \infty$. The parameterization by $N$ is meant to denote the width of the localization region where the elements $u_n^\ell(N)$ are close to the fixed point $u^*$, as shown in the localized lattice structures in Figure~\ref{fig:IntroFig}(c,d,f). We have the following important points related to the spectrum of \eqref{StabilityMap} about a localized solution:
\begin{enumerate}
    \item The set $\Sigma^0$ belongs to the essential spectrum since $u_n^\ell(N) \to 0$ as $n \to \infty$;
    \item Based on the work in the continuous setting \cite{sandstede2000absolute}, one expects that $\Sigma^*$ breaks up into $\mathcal{O}(N)$ eigenvalues;
    \item For any $\lambda \in \Omega$, \eqref{StabilityMap} about a localized solution has Fredholm index zero, meaning that only point spectra can lie in $\Omega$. 
\end{enumerate}
The following theorem characterizes the point spectrum of localized solutions based on our understanding of the point spectra of the front and back solutions.

\begin{thm}\label{thm:SinglePulse}
    Assume Hypotheses~\ref{hyp:FixedPoints}, \ref{hyp:FrontBack}, \ref{hyp:Morse}, and \ref{hyp:Localized}. Fix $\lambda_* \in \Omega$ \rev{and let $U\subset\Omega$ be a simply connected open neighbourhood of $\lambda_*$, with $\overline U\subset\Omega$. Choose $\delta>0$ sufficiently small so that
    \[
        \overline{B_\delta(\lambda_*)}\subset U.
    \]
    }
    Suppose that for $c_f,c_b \neq 0$, $m_f,m_b \geq 0$, we have
    \begin{equation}
        \begin{split}
            \rev{D_f^U}(\lambda) &= c_f(\lambda - \lambda_*)^{m_f} + \mathcal{O}(|\lambda - \lambda_*|^{m_f+1}), \\
            \rev{D_b^U}(\lambda) &= c_b(\lambda - \lambda_*)^{m_b} + \mathcal{O}(|\lambda - \lambda_*|^{m_b+1}),
        \end{split}
    \end{equation}
    for $\lambda \in B_\delta(\lambda_*)$, the ball in $\mathbb{C}$ of radius $\delta$ centered at $\lambda_*$. Then, there exists an analytic function $\rev{D_{\ell,N}^U}(\lambda)$ \rev{on $B_\delta(\lambda_*)$} such that there exists a $\delta_* > 0$ sufficiently small, with $0 < \delta_* < \delta$, and an $N_* \geq 1$ sufficiently large, such that the following holds uniformly in $N \geq N_*$:
    \begin{enumerate}
        \item $\rev{D_{\ell,N}^U}$ has precisely $m_f + m_b$ roots, counted with multiplicity, in $B_{\delta_*}(\lambda_*)$. These values of $\lambda$ are $\mathcal{O}(\mathrm{e}^{-\rev{\beta} N})$ close to $\lambda_*$, with $\rev{\beta} > 0$.
        \item The system $v_{n+1} = [DF(u_n^\ell(N)) + B(u_n^\ell(N),\lambda)]v_n$ has a nontrivial bounded solution if and only if $\rev{D_{\ell,N}^U}(\lambda) = 0$, \rev{for $\lambda\in B_{\delta_*}(\lambda_*)$}.
       \item We have $\rev{D_{\ell,N}^U}(\lambda) = D_f(\lambda)D_b(\lambda) + \mathcal{O}(\mathrm{e}^{-\alpha N})$, \rev{uniformly for $\lambda$ in compact subsets of $B_{\delta_*}(\lambda_*)$}, with \rev{$\alpha > 0$}.
    \end{enumerate}
\end{thm}

%%%%%%%%%%%%%%%%%%%%%%%%%%%%%%%%%%%%%%%%%%%%%%%%%%%%%%%%%%%%%%%%%%%%%%%%%%%%%%%%%%%%%%%%%%%%%%%%%%%%
\subsection{Extension to multi-pulses}\label{sec:MultiResult}

The previous subsection focused on the stability of localized solutions with a single localization region, so-called {\em single-pulse solutions}, parametrized by the value $N$ in Hypothesis~\ref{hyp:Localized}. The constructive proofs in \cite{bramburger2020spatially} can be used to verify this hypothesis for lattice systems of the form \eqref{LDS}. Moreover, the work in \cite{bramburger2021isolas} proves that under the same hypotheses which guarantee the existence of single-pulse solutions, one can further prove the existence of {\em multi-pulse solutions} or $k$-pulse solutions as well. These multi-pulses have $k\geq 2$ distinct regions of localization/activation where the trajectory of \eqref{Fmap} is near $u^*$, separated by long stretches where the pattern is close to $0$; see the example in Figure~\ref{fig:IntroFig}(e). All such multi-pulse patterns are amenable to our stability analysis as well. We demonstrate with the following hypothesis and stability theorem for a 2-pulse solution.  

\begin{hyp}\label{hyp:2Localized}
    There exists a family of bounded trajectories $\{u_n^{2\ell}(N_1,N_2,M_1)\}_{n \in \mathbb{Z}}$ to \eqref{Fmap}, parametrized by $N_1,N_2 \geq N_*$ and $M_1 \geq M_*$ for some $N_*,M_* \geq 1$, satisfying 
    \begin{equation}\label{2Pulse}
        u^{2\ell}_n(N_1,N_2,M_1) = \begin{cases}
            u_{n + 2N_1 + M_1}^f + w^{f,1,-}_{n + 2N_1 + M_1}(N_1,N_2,M_1) & n \leq - 2N_1 - M_1 \\
            u_{n+ 2N_1 + M_1}^f + w^{f,1,+}_{n+2N_1 + M_1}(N_1,N_2,M_1) & - 2N_1 - M_1 < n \leq - N_1 - M_1 \\
            u_{n + M_1}^b + w^{b,1,-}_{n+M_1}(N_1,N_2,M_1) & - N_1 - M_1 < n \leq -M_1 \\
            u_{n+M_1}^b + w^{b,1,+}_{n + M_1}(N_1,N_2,M_1) & -M_1 < n \leq 0  \\
            u_{n - M_1}^f + w^{f,2,-}_{n - M_1}(N_1,N_2,M_1) & 0 < n \leq M_1 \\
            u_{n - M_1}^f + w^{f,2,+}_{n - M_1}(N_1,N_2,M_1) & M_1 < n \leq M_1 + N_2 \\
            u_{n - 2N_2 - M_1}^b + w^{b,2,-}_{n- 2N_2 - M_1}(N_1,N_2,M_1) & M_1 + N_2 < n \leq M_1 + 2N_2 \\
            u_{n - 2N_2 - M_1}^b + w^{b,2,+}_{n- 2N_2 - M_1}(N_1,N_2,M_1) & M_1 + 2N_2 < n,        \end{cases}
    \end{equation}
    and moreover there exist $C,\alpha > 0$  so that
    \begin{equation}
        \begin{split}
            &|w^{f,,j,-}_{n}(N_1,N_2,M_1)| \leq C\mathrm{e}^{-\alpha(N-n)}, \qquad n\leq 0 \\
            &|w^{f,j,+}_{n}(N_1,N_2,M_1)| \leq C\mathrm{e}^{-\alpha N}, \qquad 0< n \leq N_j \\
            &|w^{b,j,-}_{n}(N_1,N_2,M_1)| \leq C\mathrm{e}^{-\alpha N}, \qquad -N_j< n \leq 0 \\
            &|w^{b,j,+}_{n}(N_1,N_2,M_1)| \leq C\mathrm{e}^{-\alpha (N+n)}, \qquad n> 0. 
        \end{split}
    \end{equation}
    for $j = 1,2$ and $N = \min\{N_1,N_2,M_1\}$.
\end{hyp}

Hypothesis~\ref{hyp:2Localized} describes a multi-pulse solution in which there are two distinct regions of localization whose length is parametrized by $N_1$ and $N_2$, separated by a stretch parametrized by $M_1$. We now provide the following stability result for the 2-pulse solution described above. The proof of this result is left to Section~\ref{sec:2PulseProof}.

\begin{thm}\label{thm:2Pulse}
    Assume Hypotheses~\ref{hyp:FixedPoints}, \ref{hyp:FrontBack}, \ref{hyp:Morse}, and \ref{hyp:2Localized}. Fix $\lambda_* \in \Omega$ \rev{and let $U\subset\Omega$ be a simply connected open neighbourhood of $\lambda_*$, with $\overline U\subset\Omega$. Choose $\delta>0$ sufficiently small so that
    \[
        \overline{B_\delta(\lambda_*)}\subset U.
    \]
    }
    Suppose that for $c_f,c_b \neq 0$, $m_f,m_b \geq 0$, we have
    \begin{equation} \label{DfDbexpansions2Pulse}
        \begin{split}
            \rev{D_f^U}(\lambda) &= c_f(\lambda - \lambda_*)^{m_f} + \mathcal{O}(|\lambda - \lambda_*|^{m_f+1}), \\
            \rev{D_b^U}(\lambda) &= c_b(\lambda - \lambda_*)^{m_b} + \mathcal{O}(|\lambda - \lambda_*|^{m_b+1}),
        \end{split}
    \end{equation}
    for $\lambda \in B_\delta(\lambda_*)$, the ball in $\mathbb{C}$ of radius $\delta$ centered at $\lambda_*$. Then, there exists an analytic function $\rev{D_{2\ell,N_1,N_2,M_1}^U}(\lambda)$ \rev{on $B_\delta(\lambda_*)$} such that there exists a $\delta_* > 0$ sufficiently small, with $0 < \delta_* < \delta$, and an $N_*,M_* \geq 1$ sufficiently large, such that the following holds uniformly in $N_1,N_2 \geq N_*$ and $M_1 \geq M_*$:
    \begin{enumerate}
        \item $\rev{D_{2\ell,N_1,N_2,M_1}^U}$ has precisely $2m_f + 2m_b$ roots, counted with multiplicity, in $B_{\delta_*}(\lambda_*)$. These values of $\lambda$ are $\mathcal{O}(\mathrm{e}^{-\rev{\beta} \min\{N_1,N_2,M_1\}})$ close to $\lambda_*$, with $\rev{\beta} > 0$.
        \item The system $v_{n+1} = [DF(u_n^{2\ell}(N_1,N_2,M_1)) + B(u_n^{2\ell}(N_1,N_2,M_1),\lambda)]v_n$ has a nontrivial bounded solution if and only if $\rev{D_{2\ell,N_1,N_2,M_1}^U}(\lambda) = 0$, \rev{for $\lambda\in B_{\delta_*}(\lambda_*)$}.
       \item We have $\rev{D_{2\ell,N_1,N_2,M_1}^U}(\lambda) = [D_f(\lambda)D_b(\lambda)]^2 + \mathcal{O}(\mathrm{e}^{-\alpha \min\{N_1,N_2,M_1\}})$, \rev{uniformly for $\lambda$ in compact subsets of $B_{\delta_*}(\lambda_*)$}, with $\alpha > 0$.
    \end{enumerate}   
\end{thm}

The above stability result essentially gives that the 2-pulse solution described in Hypothesis~\ref{hyp:2Localized} inherits $2m_f$ eigenvalues from the front and $2m_b$ from the back. This comes from the fact that a 2-pulse solution is comprised of two front and two back solutions. We can easily generalize this result to $k$-pulse solutions with $k \geq 2$, but we refrain from doing so here because the notation becomes cumbersome. Loosely, we can suppose that $\{u_n^{k\ell}(N_1,\dots,N_k,M_1,\dots,M_{k-1})\}_{n \in \mathbb{Z}}$ is a $k$-pulse solution that generalizes \eqref{2Pulse} to have $k\geq 2$ regions of localization parametrized by $N_1,\dots,N_k$, each separated by stretches parametrized by the $M_1,\dots,M_{k-1}$. Then, an analogous result to Theorem~\ref{thm:2Pulse} would show that this $k$-pulse inherits $k m_f$ eigenvalues from the front and $k m_b$ from the back. Moreover, one may further generalize this to hypothesize $k$-pulse solutions comprised of different front and back solutions, in which case the $k$-pulse would inherit eigenvalues from each of the different fronts and backs with multiplicity. Doing so becomes more of an exercise in book-keeping to keep up with a cascading notation of super- and subscripts, as was already demonstrated in moving from Hypothesis~\ref{hyp:Localized} to Hypothesis~\ref{hyp:2Localized}, while the proofs are largely the same as those for our presented stability results for single- and 2-pulse solutions.

%%%%%%%%%%%%%%%%%%%%%%%%%%%%%%%%%%%%%%%%%%%%%%%%%%%%%%%%%%%%%%%%%%%%%%%%%%%%%%%%%%%%%%%%%%%%%%%%%%%%
\section{Proofs}\label{sec:Proofs}

The proof of Theorem~\ref{thm:SinglePulse} follows similarly to \cite[Theorem~3]{makrides2019existence} with differences coming from the fact that we are now in discrete time. The proof of Theorem~\ref{thm:2Pulse} then extends these proofs to understand the stability of 2-pulse solutions. From here, the road map for further generalization to study the stability of $k$-pulses becomes clear.

%%%%%%%%%%%%%%%%%%%%%%%%%%%%%%%%%%%%%%%%%%%%%%%%%%%%%%%%%%%%%%%%%%%%%%%%%%%%%%%%%%%%%%%%%%%%%%%%%%%%
\subsection{Proof of Theorem~\ref{thm:SinglePulse}}\label{sec:1PulseProof}

We begin by recalling from Hypothesis~\ref{hyp:Localized} that our localized solution is given by 
\begin{equation}
    u^\ell_n(N) = \begin{cases}
        u_{n+ N}^f + w^{f,-}_{n+N}(N) & n \leq -N \\
        u_{n+N}^f + w^{f,+}_{n+N}(N) & -N < n \leq 0 \\
        u_{n-N}^b + w^{b,-}_{n-N}(N) & 0 < n \leq N \\
        u_{n-N}^b + w^{b,+}_{n-N}(N) & N < n, 
    \end{cases}
\end{equation}
and parametrized by a sufficiently large $N \geq 1$. Evaluating the eigenvalue problem \eqref{StabilityMap} at this localized solution allows one to split $v_{n+1} = [DF(u_n^\ell(N)) + B(u_n^\ell(N),\lambda)]v_n$ into pieces as
\begin{equation}\label{1PulseProof1}
    \begin{cases}
        v^{f,-}_{n+1} = [DF(u_{n }^f + w^{f,-}_{n}(N)) + B(u_{n}^f + w^{f,-}_{n}(N),\lambda)]v^{f,-}_n, & n \leq 0 \\
        v^{f,+}_{n+1} = [DF(u_{n }^f + w^{f,+}_{n}(N)) + B(u_{n}^f + w^{f,+}_{n}(N),\lambda)]v^{f,+}_n, & 0 \leq n \leq N \\
        v^{b,-}_{n+1} = [DF(u_{n }^b + w^{b,-}_{n}(N)) + B(u_{n}^b + w^{b,-}_{n}(N),\lambda)]v^{b,-}_n, & -N \leq n \leq 0 \\
        v^{b,+}_{n+1} = [DF(u_{n }^b + w^{b,+}_{n}(N)) + B(u_{n}^b + w^{b,+}_{n}(N),\lambda)]v^{b,+}_n, &  n \geq 0. 
    \end{cases}
\end{equation}
Thus, a solution of the above is constructed as 
\begin{equation}\label{1PulseProofMatch}
    v_n = \begin{cases}
        v^{f,-}_{n+N}, & n \leq -N \\
        v^{f,+}_{n+N}, & -N \leq n \leq 0 \\
        v^{b,-}_{n-N}, & 0 \leq n \leq N \\
        v^{b,+}_{n-N}, & n \geq N, 
    \end{cases}
\end{equation}
along with the matching conditions 
\begin{subequations} \label{1PulseMatch}
    \begin{align}
        v^{f,+}_{N} - v^{b,-}_{-N} & = 0,  \label{1PulseMatch1}\\
        v^{f,+}_{0} - v^{f,-}_{0} & = 0,  \label{1PulseMatch2}\\
        v^{b,+}_{0} - v^{b,-}_{0} & = 0. \label{1PulseMatch3}
    \end{align}
\end{subequations}

Now, let us fix $\lambda_* \in \Omega$ as in the statement of Theorem~\ref{thm:SinglePulse}. From Lemma~\ref{lem:FrontBackEssential}, the systems \eqref{FrontBackLinear} have exponential dichotomies on $\mathbb{Z}_-$ and $\mathbb{Z}_+$. Combining this with the roughness theorem for exponential dichotomies \cite[Proposition~2.5]{beyn1997numerical}, we have exponential dichotomies for each difference equation in \eqref{1PulseProof1}, which can be chosen so that they depend analytically on $\lambda \in B_\delta(\lambda_*)$ for $\delta > 0$ small. For simplicity, we denote the solution operator composed with the projections as $\Phi^s(n,m) = \Phi^s(n,m)P(m)$ and $\Phi^u(n,m) = \Phi^u(n,m)(I - P(m))$, giving that there exists $C,\alpha > 0$ so that the following bounds hold: 
\begin{equation}\label{1PulseDichotomyDecay}
    \begin{split}
		n,m \leq 0 \quad &\begin{cases}
			|\Phi_\ell^{f,-,s}(n,m;\lambda)| \leq C\mathrm{e}^{-\alpha (n - m)}, & n \geq m \\
            |\Phi_\ell^{f,-,u}(n,m;\lambda)| \leq C\mathrm{e}^{-\alpha (m - n)}, & n \leq m \\
		\end{cases}
			\\
		0\leq n,m \leq N \quad &\begin{cases}
			|\Phi_\ell^{f,+,s}(n,m;\lambda)| \leq C\mathrm{e}^{-\alpha (n - m)}, & n \geq m \\
            |\Phi_\ell^{f,+,u}(n,m;\lambda)| \leq C\mathrm{e}^{-\alpha (m - n)}, & n \leq m \\
		\end{cases}
			\\
        -N\leq n,m \leq 0 \quad &\begin{cases}
			|\Phi_\ell^{b,-,s}(n,m;\lambda)| \leq C\mathrm{e}^{-\alpha (n - m)}, & n \geq m \\
            |\Phi_\ell^{b,-,u}(n,m;\lambda)| \leq C\mathrm{e}^{-\alpha (m - n)}, & n \leq m \\
		\end{cases}
			\\
        n,m \geq 0 \quad &\begin{cases}
			|\Phi_\ell^{b,+,s}(n,m;\lambda)| \leq C\mathrm{e}^{-\alpha (n - m)}, & n \geq m \\
            |\Phi_\ell^{b,+,u}(n,m;\lambda)| \leq C\mathrm{e}^{-\alpha (m - n)}, & n \leq m \\
		\end{cases}
	   \end{split}
\end{equation}
\rev{The constants $C,\alpha > 0$ in \eqref{1PulseDichotomyDecay} can be chosen independently of $N$ for all sufficiently large $N$, uniformly for $\lambda$ in a fixed compact neighbourhood of $\lambda_*$ contained in $B_\delta(\lambda_*)$. Indeed, the unperturbed front and back dichotomies are independent of $N$, while Hypothesis~\ref{hyp:Localized} gives coefficient perturbations of size $\mathcal{O}(\mathrm{e}^{-\alpha_0 N})$ for some $\alpha_0 > 0$. The roughness theorem yields dichotomy bounds uniform in $N$.}

From the above notation, the associated projections are given by $P_\ell^{i,\pm,j}(n;\lambda) = \Phi_\ell^{i,\pm,j}(n,n;\lambda)$ where $i = f,b$ and $j = s,u$. Moreover, the roughness theorem gives that 
\begin{equation}\label{1PulseProofConvergingProjections}
    \begin{split}
        |P_\ell^{f,+,u}(N;\lambda) - P_*^u(0;\lambda)| &\leq C\mathrm{e}^{-\alpha N}, \\
        |P_\ell^{b,-,s}(-N;\lambda) - P_*^s(0;\lambda)| &\leq C\mathrm{e}^{-\alpha N}, 
    \end{split}
\end{equation}
where $P_*^s$ is the projection guaranteed by Definition~\ref{def:ExpDich} for $v_{n+1} = [DF(u^*) +  B(u^*,\lambda)]v_n$ in \eqref{0starDiffEqns} and $P_*^s = I - P_*^u$. We now let $a := (a^{f,+},a^{b,-}) \in V_a$, $b := (b^{f,-},b^{f,+},b^{b,-},b^{b,+}) \in V_b$, and $\lambda \in V_\lambda$, where the spaces $V_a$, $V_b$, and $V_\lambda$ are defined as follows:
\begin{equation}\label{VaVbSpaces}
    \begin{split}
        V_a &:= \mathrm{Rg}(P_*^u(0;\lambda_*))\oplus\mathrm{Rg}(P_*^s(0;\lambda_*)) \\
        V_b &:= \bigg(\mathrm{Rg}(P_\ell^{f,-,u}(0;\lambda_*))\oplus\mathrm{Rg}(P_\ell^{f,+,s}(0;\lambda_*))\bigg)\oplus\bigg(\mathrm{Rg}(P_\ell^{b,-,u}(0;\lambda_*))\oplus\mathrm{Rg}(P_\ell^{b,+,s}(0;\lambda_*))\bigg) \\
        V_\lambda &:= B_\delta(\lambda_*) \subset \Omega,
    \end{split}
\end{equation}
where $V_a$ and $V_b$ are endowed with the maximum norm over each of the elements. 

For $\delta > 0$ sufficiently small and $N \geq N_*$ sufficiently large, we can then write solutions to the eigenvalue problem \eqref{1PulseProof1} for the localized solution $\{u_n^\ell\}_{n\in\mathbb{Z}}$ as
\begin{equation}\label{1PulseProof2}
    \begin{split}
        v_n^{f,-} &= \Phi_\ell^{f,-,u}(n,0;\lambda)b^{f,-}, \quad n \leq 0 \\
        v_n^{f,+} &= \Phi_\ell^{f,+,s}(n,0;\lambda)b^{f,+} + \Phi_\ell^{f,+,u}(n,N;\lambda)a^{f,+}, \quad 0 \leq n \leq N \\
        v_n^{b,-} &= \Phi_\ell^{b,-,s}(n,-N;\lambda)a^{b,-} + \Phi_\ell^{b,-,u}(n,0;\lambda)b^{b,-}, \quad -N \leq n \leq 0 \\
        v_n^{b,+} &= \Phi_\ell^{b,+,s}(n,0;\lambda)b^{b,+}, \quad n \geq 0.
    \end{split}
\end{equation}
Having nonzero $a,b$ in the above and satisfying the matching conditions \eqref{1PulseMatch} will result in a nontrivial bounded solution to \eqref{1PulseProof1}. Notice that the analytic dependence of the projections on $\lambda$ implies that (i) if $\delta > 0$ is sufficiently small, then the range of $P_\ell^{i,\pm,j}(0;\lambda)$ has no nontrivial overlap with the kernel of $P_\ell^{i,\pm,j}(0;\lambda_*)$ for all $\lambda$ near $\lambda_*$, $i = f,b$, and $j = s,u$. Similarly, for $N$ sufficiently large the range of $P_\ell^{f,+,u}(N;\lambda)$ has no nontrivial overlap with the kernel of $P_*^u(0;\lambda_*)$, and similarly for the range of $P_\ell^{b,-,s}(-N;\lambda)$ and the kernel of $P_*^s(0;\lambda_*)$, owing to the estimates \eqref{1PulseProofConvergingProjections}. Thus, for $\delta$ sufficiently small and $N$ sufficiently large, it follows that all bounded solutions of \eqref{1PulseProof1} take the form \eqref{1PulseProof2} and satisfy the matching conditions \eqref{1PulseMatch}.

The goal in what follows is to use the matching conditions \eqref{1PulseProofMatch} to solve for $a$ and $b$. We have the following lemma to solve for $a$ in terms of $b$.

\begin{lem}\label{lem:1PulseMatch1}
    There exists an $N_* \geq 1$ such that for all $N \geq N_*$ the following holds uniformly in $N$. There exists an operator $G:V_\lambda\times V_b \to V_a$ such that $v_n$ as given in \eqref{1PulseProof2} with $a = G(\lambda)b$ solves \eqref{1PulseMatch1} for any $b$ and $\lambda$. The function $G$ is analytic in $\lambda$ and linear in $b$, and satisfies 
    \begin{equation}\label{Gestimate}
        |G(\lambda)b|\leq C\mathrm{e}^{-\alpha N}|b|.
    \end{equation}
\end{lem}

\begin{proof}
    Substituting the solution form \eqref{1PulseProof2} into the condition in \eqref{1PulseMatch1} gives
    \begin{equation}\label{1PulseProof3}
        \begin{split}
            0 &= \Phi_\ell^{f,+,s}(N,0;\lambda)b^{f,+} + \Phi_\ell^{f,+,u}(N,N;\lambda)a^{f,+} - \Phi_\ell^{b,-,s}(-N,-N;\lambda)a^{b,-} + \Phi_\ell^{b,-,u}(-N,0;\lambda)b^{b,-} \\
            &=\bigg(P_\ell^{f,+,u}(N;\lambda) - P_*^{u}(0;\lambda)\bigg)a^{f,+} + a^{f,+} + \bigg(P_*^{s}(0;\lambda) - P_\ell^{b,-,s}(N;\lambda)\bigg)a^{b,-} - a^{b,-} \\
            &\qquad + \Phi_\ell^{f,+,s}(N,0;\lambda)b^{f,+} + \Phi_\ell^{b,-,u}(-N,0;\lambda)b^{b,-}.
        \end{split}
    \end{equation}
Now, let us define 
\begin{equation}
    \begin{split}
    \tilde{G}(\lambda)(a,b) &= \bigg(P_\ell^{f,+,u}(N;\lambda) - P_*^{u}(0;\lambda)\bigg)a^{f,+} + \bigg(P_*^{s}(0;\lambda) - P_\ell^{b,-,s}(N;\lambda)\bigg)a^{b,-} \\
          \\ &\qquad\qquad  + \Phi_\ell^{f,+,s}(N,0;\lambda)b^{f,+} + \Phi_\ell^{b,-,u}(-N,0;\lambda)b^{b,-}.
    \end{split}
\end{equation}
Note that $\tilde{G}(\lambda)(a,b)$ is analytic in $\lambda$ because all projections and solution operators involved are analytic in $\lambda$. Moreover, we can see that $\tilde{G}$ is linear in both $a$ and $b$. From \eqref{1PulseDichotomyDecay} and \eqref{1PulseProofConvergingProjections} we get the estimate
\begin{equation}\label{1PulseProof4}
    |\tilde{G}(\lambda)(a,b)| \leq C\mathrm{e}^{-\alpha N}(|a| + |b|).
\end{equation}

Now, define the map $K: V_a \to \mathbb{C}^d$ by $K(a^{f,+},a^{b,-}) = a^{f,+} - a^{b,-}$. Since $\lambda_* \in \Omega$, Lemma~\ref{lem:FrontBackEssential} gives that $V_a = \mathrm{Rg}(P_*^u(0;\lambda_*))\oplus\mathrm{Rg}(P_*^s(0;\lambda_*)) = \mathbb{C}^d$, and so $K$ is a bounded linear isomorphism. Notice that \eqref{1PulseProof3} can equivalently be written as 
\begin{equation}
    0 = K(a) + \tilde{G}(\lambda)(a,0) + \tilde{G}(\lambda)(0,b) \implies (K + \tilde{G}(\lambda)J)a = -\tilde{G}(\lambda)(0,b),
\end{equation}
where $J(a) := (a,0)$. From the estimate \eqref{1PulseProof4}, it follows that for $N \geq N_*$ sufficiently large, $(K + \tilde{G}(\lambda)J)$ is invertible and so we can solve for $a$ to obtain 
\begin{equation}
    a = -(K + \tilde{G}(\lambda)J)^{-1}\tilde{G}(\lambda)(0,b) =: G(\lambda)b,
\end{equation}
giving the function defined in the lemma. The estimate \eqref{Gestimate} follows directly from the boundedness of $(K + \tilde{G}(\lambda)J)^{-1}$ and \eqref{1PulseProof4}, completing the proof.  
\end{proof} % end of proof

With Lemma~\ref{lem:1PulseMatch1}, it now remains to solve the remaining matching equations \eqref{1PulseMatch2} and \eqref{1PulseMatch3}. First, substituting \eqref{1PulseProof2} into the matching condition $v^{f,+}_{0} - v^{f,-}_{0}  = 0$ gives
\begin{equation}\label{1PulseMatch2v2}
    \begin{split}
        0 &= \Phi_\ell^{f,+,s}(0,0;\lambda)b^{f,+} + \Phi_\ell^{f,+,u}(0,N;\lambda)a^{f,+} - \Phi_\ell^{f,-,u}b^{f,-} \\
        &= P_\ell^{f,+,s}(0;\lambda)b^{f,+} - P_\ell^{f,-,u}(0;\lambda)\rev{b^{f,-}} + \Phi_\ell^{f,+,u}(0,N;\lambda)(G(\lambda)b)^{f,+},
    \end{split}
\end{equation}
where we have replaced $a^{f,+}$ in the second line with the function $G(\lambda)b$ from Lemma~\ref{lem:1PulseMatch1}, i.e. $(G(\lambda)b)^{f,+}$ is the component of $G(\lambda)b$ in $\mathrm{Rg}(P_*^u(0;\lambda))$. A nearly identical chain of reasoning allows us to write the third matching condition $v^{b,+}_{0} - v^{b,-}_{0}  = 0$ as 
\begin{equation}\label{1PulseMatch3v2}
    0 = P_\ell^{b,+,s}(0;\lambda)b^{b,+} - \rev{P_\ell^{b,-,u}}(0;\lambda)b^{b,-} - \Phi_\ell^{b,-,s}(0,-N;\lambda)(G(\lambda)b)^{b,-},
\end{equation}
where $(G(\lambda)b)^{b,-}$ is the component of $(G(\lambda)b)$ in $\mathrm{Rg}(P_*^s(0;\lambda))$. From the estimates \eqref{1PulseDichotomyDecay} and Lemma~\ref{lem:1PulseMatch1} we have that 
\begin{equation}\label{1PulseProof5}
    \begin{split}
        \Phi_\ell^{f,+,u}(0,N;\lambda)(G(\lambda)b)^{f,+} &= \mathcal{O}(\mathrm{e}^{-2\alpha N}|b|),\\
        \Phi_\ell^{b,-,s}(0,-N;\lambda)(G(\lambda)b)^{b,-} &= \mathcal{O}(\mathrm{e}^{-2\alpha N}|b|),   
    \end{split}
\end{equation}
uniformly in $\lambda$ near $\lambda_*$. We combine the linear matching equations \eqref{1PulseMatch2v2} and \eqref{1PulseMatch3v2} together in matrix form as 
\begin{equation}
\begin{split}
    \begin{bmatrix}
        0 \\ 0
    \end{bmatrix} &= \Bigg(\begin{bmatrix}
        - \rev{P_\ell^{f,-,u}(0;\lambda)} & P_\ell^{f,+,s}(0;\lambda) & 0 & 0 \\
        0 & 0 & - \rev{P_\ell^{b,-,u}(0;\lambda)} & P_\ell^{b,+,s}(0;\lambda)
    \end{bmatrix} \\ 
    & \qquad + \begin{bmatrix}
       \Phi_\ell^{f,+,u}(0,N;\lambda)G(\lambda) & 0 \\ 0 & -\Phi_\ell^{b,-,s}(0,-N;\lambda)G(\lambda)   
    \end{bmatrix}\Bigg)b,
\end{split}
\end{equation}
which we can write in compact form as $0 = [P_N(\lambda) + R_N(\lambda)]b$. Notice that these matrices are square since $b \in V_b$ with $\mathrm{dim}(V_b) = i_\infty + (d - i_\infty) + i_\infty + (d - i_\infty) = 2d$, according to Hypothesis~\ref{hyp:Morse}, meaning $P_N(\lambda),R_N(\lambda): \mathbb{C}^{2d} \to \mathbb{C}^{2d}$. Furthermore, from \eqref{1PulseProof5} we have that $R_N(\lambda) = \mathcal{O}(\mathrm{e}^{-2\alpha N})$ uniformly in $\lambda$ near $\lambda_*$.

Now, let us define $\hat{P}_\ell^{f,-,u}(0;\lambda)$ to be $P_\ell^{f,-,u}(0;\lambda)$ restricted to $\mathrm{Rg}(P_\ell^{f,-,u}(0;\lambda_*))$ and similarly for the other projections $P_\ell^{f,+,s}(0;\lambda)$, $P_\ell^{b,-,u}(0;\lambda)$, and $P_\ell^{b,+,s}(0;\lambda)$. Then, we define 
\begin{equation}\label{Phat}
    \hat{P}_N(\lambda) = \begin{bmatrix}
        - \rev{\hat{P}_\ell^{f,-,u}}(0;\lambda) & \hat{P}_\ell^{f,+,s}(0;\lambda) & 0 & 0 \\
        0 & 0 & - \rev{\hat{P}_\ell^{b,-,u}}(0;\lambda) & \hat{P}_\ell^{b,+,s}(0;\lambda)
    \end{bmatrix} 
\end{equation}
and $\hat{R}_N(\lambda)$ as $R_N(\lambda)$ restricted to $V_b$. \rev{We fix bases of the finite-dimensional spaces comprising $V_b$, and all determinants of operators acting on $V_b$ below are understood with respect to these fixed bases.} This leads to the definition 
\begin{equation}
    D_{\ell,N}(\lambda) := \mathrm{det}(\hat{P}_N(\lambda) + \hat{R}_N(\lambda))
\end{equation}
so that $v_{n+1} = [DF(u_n^\ell(N)) + B(u_n^\ell(N),\lambda)]v_n$ has a nontrivial bounded solution at $\lambda$ near $\lambda_*$ if and only if $D_{\ell,N}(\lambda) = 0$. Since $\hat{P}_N(\lambda)$ and $\hat{R}_N(\lambda)$ are analytic in $\lambda$ and since $\hat{R}_N(\lambda) = \mathcal{O}(\mathrm{e}^{-2\alpha N})$, if $\lambda$ is sufficiently close to $\lambda_*$ and $N$ is sufficiently large we then have that
\begin{equation}
    D_{\ell,N}(\lambda) = \mathrm{det}(\hat{P}_N(\lambda)) + \mathcal{O}(\mathrm{e}^{-2\alpha N}). 
\end{equation}

Now, the block diagonal structure of $\hat{P}_N(\lambda)$ gives that
\begin{equation}
    \mathrm{det}(\hat{P}_N(\lambda)) = \mathrm{det}\bigg(- \rev{\hat{P}_\ell^{f,-,u}}(0;\lambda)\ \ \hat{P}_\ell^{f,+,s}(0;\lambda)\bigg)\cdot \mathrm{det} \bigg(- \rev{\hat{P}_\ell^{b,-,u}}(0;\lambda)\ \ \hat{P}_\ell^{b,+,s}(0;\lambda)\bigg),
\end{equation}
where we recall that the hats indicate we are restricting to the range of each projection at $\lambda = \lambda_*$. \rev{We choose the perturbed frames compatibly with the unperturbed frames used in \eqref{EvansFrontBack}. Specifically, set
\[
    Z_{\ell,N}^{i,\pm,j}(\lambda)
    :=P_\ell^{i,\pm,j}(0;\lambda)Z_i^{\pm,j}(\lambda),
\]
for $(i,\pm,j)=(f,-,u),(f,+,s),(b,-,u),(b,+,s)$. Since $P_i^{\pm,j}(0;\lambda)Z_i^{\pm,j}(\lambda) =Z_i^{\pm,j}(\lambda)$, closeness of the projections gives
\[
    Z_{\ell,N}^{i,\pm,j}(\lambda)-Z_i^{\pm,j}(\lambda)
    =\mathcal{O}(\mathrm{e}^{-\alpha N}),
\]
uniformly on a fixed compact neighbourhood of $\lambda_*$. For sufficiently large $N$, these matrices therefore have full column rank and form analytic frames of the perturbed subspaces.

Choose the fixed bases of the components of $V_b$ to be the columns of $Z_{\ell,N}^{i,\pm,j}(\lambda_*)$. After shrinking the neighbourhood of $\lambda_*$ if necessary, the changes of coordinates from the projected fixed bases to the compatible frames above are analytic and uniformly bounded, together with their inverses, for all sufficiently large $N$. Their combined determinant gives an analytic nonvanishing factor $h_N(\lambda)$, with $h_N$ and $h_N^{-1}$ uniformly bounded. Defining
\[
    D_{\ell,N}^U(\lambda):=h_N(\lambda)D_{\ell,N}(\lambda),
\]
we obtain
\begin{equation}\label{eq:Evans ranges 1pulse}
\begin{split}
    D_{\ell,N}^U(\lambda)
    =
    \det\!\left(
        Z_{\ell,N}^{f,-,u}(\lambda)\ \
        Z_{\ell,N}^{f,+,s}(\lambda)
    \right)
    \times
    \det\!\left(
        Z_{\ell,N}^{b,-,u}(\lambda)\ \
        Z_{\ell,N}^{b,+,s}(\lambda)
    \right)
    +\mathcal{O}(\mathrm{e}^{-2\alpha N}).
\end{split}
\end{equation}
This normalization preserves the zeros and their multiplicities. The frame estimates then yield
\[
    D_{\ell,N}^U(\lambda)
    =D_f^U(\lambda)D_b^U(\lambda)
      +\mathcal{O}(\mathrm{e}^{-\alpha N}),
\]
uniformly on the chosen compact neighbourhood.
}

Now, since the projections satisfy $|P_\ell^{f,-,u}(0;\lambda) - P_{f}^{-,u}(0;\lambda)| \leq C\mathrm{e}^{-\alpha N}$, and analogously for all other projections in $\rev{D_{\ell,N}^U}$ above, for some uniform constant $C > 0$ independent of $\lambda$ near $\lambda_*$ and $N$ sufficiently large. Thus, this gives that $\rev{D_{\ell,N}^U}(\lambda)$ can be written as
\begin{equation}
    \rev{D_{\ell,N}^U} = (\rev{D_f^U}(\lambda) + \mathcal{O}(\mathrm{e}^{-\alpha N}))(\rev{D_b^U}(\lambda) + \mathcal{O}(\mathrm{e}^{-\alpha N})) + \mathcal{O}(\mathrm{e}^{-2\alpha N}),
\end{equation}
where $\rev{D_f^U}(\lambda)$ and $\rev{D_b^U}(\lambda)$ are defined in \eqref{EvansFrontBack}. 

Since we have assumed that for $\lambda \in B_\delta(\lambda_*)$ we have 
\begin{equation}
    \rev{D_f^U}(\lambda) = c_f(\lambda - \lambda_*)^{m_f} + \mathcal{O}(|\lambda - \lambda_*|^{m_f + 1})
\end{equation}
and 
\begin{equation}
    \rev{D_b^U}(\lambda) = c_b(\lambda - \lambda_*)^{m_b} + \mathcal{O}(|\lambda - \lambda_*|^{m_b + 1}),
\end{equation}
for some $\rev{m_f,m_b \geq 0}$ and $c_b,c_f \neq 0$. Now, one can see that near $\lambda_*$ we have the expansion 
\begin{equation}
    \rev{D_{\ell,N}^U}(\lambda) = c_\ell(\lambda - \lambda_*)^{m_f + m_b} + \mathcal{O}(|\lambda - \lambda_*|^{m_f + m_b + 1}) + \mathcal{O}(\mathrm{e}^{-\alpha N}),
\end{equation}
for some $c_\ell \neq 0$. To show that $\rev{D_{\ell,N}^U}$ has $(m_f + m_b)$ roots in a $N$-independent neighborhood of $\lambda_*$ for all $N$ sufficiently large, we evoke Rouch\'e's theorem. An identical argument is made to prove \cite[Lemma~6.10]{makrides2019existence} and we simply direct the reader to the proof there to complete our proof of Theorem~\ref{thm:SinglePulse}.

\rev{The local spectral approximation results yield eigenvalue counts on compact domains by a finite-cover argument. This does not require a single localized Evans function defined throughout $U$.}

\begin{cor} \label{cor:SinglePulse uniform}
    \rev{Assume Hypothesis~\ref{hyp:FixedPoints}, \ref{hyp:FrontBack}, \ref{hyp:Morse}, and \ref{hyp:Localized}. Let $G$ be a bounded domain with piecewise smooth boundary and $\overline G \subset U$. Suppose that neither $D_f^U$ nor $D_b^U$ vanishes on $\partial G$. Then, for all sufficiently large $N$, the number of eigenvalues of the localized state in $G$, counted with algebraic multiplicity, equals the sum of the corresponding counts for the front and back.}
\end{cor}

\begin{proof}
    \rev{The constituent Evans functions are analytic on $U$ and nonzero on $\partial G$, so they have only finitely many zeros in $\overline G$, all contained in $G$. Apply Theorem~\ref{thm:SinglePulse} at each distinct zero, choosing disjoint closed disks contained in $G$. For all sufficiently large $N$, each disk contains the prescribed number of localized-state eigenvalues.}

    \rev{The complement of the open disks in $\overline G$ is compact and contains no constituent eigenvalues. At each point of this complement, Theorem~\ref{thm:SinglePulse}, with $m_f=m_b=0$, provides a neighbourhood containing no localized-state eigenvalues for sufficiently large $N$. A finite subcover gives a common threshold for $N$, excluding all additional eigenvalues in $\overline G$ and completing the count.}
\end{proof}

%%%%%%%%%%%%%%%%%%%%%%%%%%%%%%%%%%%%%%%%%%%%%%%%%%%%%%%%%%%%%%%%%%%%%%%%%%%%%%%%%%%%%%%%%%%%%%%%%%%%
\subsection{Proof of Theorem~\ref{thm:2Pulse}}\label{sec:2PulseProof}

Much of this proof follows through similar arguments to that of Theorem~\ref{thm:SinglePulse} and so we only seek to sketch out the key points and highlight the differences. We recall that Hypothesis~\ref{hyp:2Localized} guarantees the existence of a 2-pulse solution $\{u^{2\ell}_n(N_1,N_2,M_1)\}_{n \in \mathbb{Z}}$ of \eqref{Fmap}, parametrized by $N_1,N_2,M_1$ sufficiently large, that takes the form
    \begin{equation}\label{2PulseProof}
        u^{2\ell}_n(N_1,N_2,M_1) = \begin{cases}
            u_{n + 2N_1 + M_1}^f + w^{f,1,-}_{n + 2N_1 + M_1}(N_1,N_2,M_1) & n \leq - 2N_1 - M_1 \\
            u_{n+ 2N_1 + M_1}^f + w^{f,1,+}_{n+2N_1 + M_1}(N_1,N_2,M_1) & - 2N_1 - M_1 < n \leq - N_1 - M_1 \\
            u_{n + M_1}^b + w^{b,1,-}_{n+M_1}(N_1,N_2,M_1) & - N_1 - M_1 < n \leq -M_1 \\
            u_{n+M_1}^b + w^{b,1,+}_{n + M_1}(N_1,N_2,M_1) & -M_1 < n \leq 0  \\
            u_{n - M_1}^f + w^{f,2,-}_{n + 2N_1 - M_1}(N_1,N_2,M_1) & 0 < n \leq M_1 \\
            u_{n - M_1}^f + w^{f,2,+}_{n - M_1}(N_1,N_2,M_1) & M_1 < n \leq M_1 + N_2 \\
            u_{n - 2N_2 - M_1}^b + w^{b,2,-}_{n- 2N_2 - M_1}(N_1,N_2,M_1) & M_1 + N_2 < n \leq M_1 + 2N_2 \\
            u_{n - 2N_2 - M_1}^b + w^{b,2,-}_{n- 2N_2 - M_1}(N_1,N_2,M_1) & M_1 + 2N_2 < n.        \end{cases}
    \end{equation}
Evaluating \eqref{StabilityMap} at this localized solution and using the above decomposition of the solution allows one to split $v_{n+1} = [DF(u_n^{2\ell}(N_1,N_2,M_1)) + B(u_n^{2\ell}(N_1,N_2,M_1),\lambda)]v_n$ into  
\begin{equation}\label{2PulseProof1}
    \begin{cases}
        v^{f,1,-}_{n+1} = [DF(u_{n }^f + w^{f,1,-}_{n}(N_1,N_2,M_1)) + B(u_{n}^f + w^{f,1,-}_{n}(N_1,N_2,M_1),\lambda)]v^{f,1,-}_n, & n \leq 0 \\
        v^{f,1,+}_{n+1} = [DF(u_{n }^f + w^{f,1,+}_{n}(N_1,N_2,M_1)) + B(u_{n}^f + w^{f,1,+}_{n}(N_1,N_2,M_1),\lambda)]v^{f,1,+}_n, & 0 \leq n \leq N_1 \\
        v^{b,1,-}_{n+1} = [DF(u_{n }^b + w^{b,1,-}_{n}(N_1,N_2,M_1)) + B(u_{n}^b + w^{b,1,-}_{n}(N_1,N_2,M_1),\lambda)]v^{b,1,-}_n, & -N_1 \leq n \leq 0 \\
        v^{b,1,+}_{n+1} = [DF(u_{n }^b + w^{b,1,+}_{n}(N_1,N_2,M_1)) + B(u_{n}^b + w^{b,1,+}_{n}(N_1,N_2,M_1),\lambda)]v^{b,1,+}_n, &  0 \leq n \leq M_1 \\
         v^{f,2,-}_{n+1} = [DF(u_{n }^f + w^{f,2,-}_{n}(N_1,N_2,M_1)) + B(u_{n}^f + w^{f,2,-}_{n}(N_1,N_2,M_1),\lambda)]v^{f,2,-}_n, & -M_1 \leq n \leq 0 \\
        v^{f,2,+}_{n+1} = [DF(u_{n }^f + w^{f,2,+}_{n}(N_1,N_2,M_1)) + B(u_{n}^f + w^{f,2,+}_{n}(N_1,N_2,M_1),\lambda)]v^{f,2,+}_n, & 0 \leq n \leq N_2 \\
        v^{b,2,-}_{n+1} = [DF(u_{n }^b + w^{b,2,-}_{n}(N_1,N_2,M_1)) + B(u_{n}^b + w^{b,2,-}_{n}(N_1,N_2,M_1),\lambda)]v^{b,2,-}_n, & -N_2 \leq n \leq 0 \\
        v^{b,2,+}_{n+1} = [DF(u_{n }^b + w^{b,2,+}_{n}(N_1,N_2,M_1)) + B(u_{n}^b + w^{b,2,+}_{n}(N_1,N_2,M_1),\lambda)]v^{b,2,+}_n, &  n \geq 0.
    \end{cases}
\end{equation}
Thus, a solution of \eqref{2PulseProof1} is constructed as 
\begin{equation}
    v_n = \begin{cases}
            v_{n + 2N_1 + M_1}^{f,1,-}, & n \leq - 2N_1 - M_1 \\
            v_{n+ 2N_1 + M_1}^{f,1,+}, & - 2N_1 - M_1 < n \leq - N_1 - M_1 \\
            v_{n + M_1}^{b,1,-}, & - N_1 - M_1 < n \leq -M_1 \\
            v_{n+M_1}^{b,1,+}, & -M_1 < n \leq 0  \\
            v_{n - M_1}^{f,2,-},  & 0 < n \leq M_1 \\
            v_{n - M_1}^{f,2,+}, & M_1 < n \leq M_1 + N_2 \\
            v_{n - 2N_2 - M_1}^{b,2,-}, & M_1 + N_2 < n \leq M_1 + 2N_2 \\
            v_{n - 2N_2 - M_1}^{b,2,+},  & M_1 + 2N_2 < n.        
    \end{cases} 
\end{equation}
along with the matching conditions
\begin{subequations}
    \begin{align}
        v^{f,1,+}_{N_1} - v^{b,1,-}_{-N_1} & = 0,  \label{2PulseMatch1}\\
        v^{f,2,+}_{N_2} - v^{b,2,-}_{-N_2} & = 0,  \label{2PulseMatch2}\\
        v^{b,1,+}_{M_1} - v^{f,2,-}_{-M_1} & = 0,  \label{2PulseMatch3}\\
        v^{f,1,+}_{0} - v^{f,1,-}_{0} & = 0,  \label{2PulseMatch4}\\
        v^{b,1,+}_{0} - v^{b,1,-}_{0} & = 0. \label{2PulseMatch5} \\
        v^{f,2,+}_{0} - v^{f,2,-}_{0} & = 0,  \label{2PulseMatch6}\\
        v^{b,2,+}_{0} - v^{b,2,-}_{0} & = 0. \label{2PulseMatch7}
    \end{align}
\end{subequations}

Now, let us fix $\lambda_* \in \Omega$ as in the statement of Theorem~\ref{thm:2Pulse}. Analogous to \eqref{1PulseDichotomyDecay}, for all $\lambda \in B_\delta(\lambda_*)$ with $\delta > 0$ sufficiently small, there exist $C,\alpha > 0$ so that  
\begin{equation}\label{2PulseDichotomyDecay}
    \begin{split}
		n,m \leq 0 \quad &\begin{cases}
			|\Phi_{2\ell,k}^{f,-,s}(n,m;\lambda)| \leq C\mathrm{e}^{-\alpha (n - m)}, & n \geq m \\
            |\Phi_{2\ell,k}^{f,-,u}(n,m;\lambda)| \leq C\mathrm{e}^{-\alpha (m - n)}, & n \leq m \\
		\end{cases}
			\\
		0\leq n,m \leq N_k \quad &\begin{cases}
			|\Phi_{2\ell,k}^{f,+,s}(n,m;\lambda)| \leq C\mathrm{e}^{-\alpha (n - m)}, & n \geq m \\
            |\Phi_{2\ell,k}^{f,+,u}(n,m;\lambda)| \leq C\mathrm{e}^{-\alpha (m - n)}, & n \leq m \\
		\end{cases}
			\\
        -N_k\leq n,m \leq 0 \quad &\begin{cases}
			|\Phi_{2\ell,k}^{b,-,s}(n,m;\lambda)| \leq C\mathrm{e}^{-\alpha (n - m)}, & n \geq m \\
            |\Phi_{2\ell,k}^{b,-,u}(n,m;\lambda)| \leq C\mathrm{e}^{-\alpha (m - n)}, & n \leq m \\
		\end{cases}
			\\
        n,m \geq 0 \quad &\begin{cases}
			|\Phi_{2\ell,k}^{b,+,s}(n,m;\lambda)| \leq C\mathrm{e}^{-\alpha (n - m)}, & n \geq m \\
            |\Phi_{2\ell,k}^{b,+,u}(n,m;\lambda)| \leq C\mathrm{e}^{-\alpha (m - n)}, & n \leq m \\
		\end{cases}
	   \end{split}
\end{equation}
for $k = 1,2$. \rev{As in the single-pulse case, the constants $C,\alpha > 0$ can be chosen independently of $N_1,N_2,M_1$ whenever $\min\{N_1,N_2,M_1\}$ is sufficiently large, uniformly for $\lambda$ in a fixed compact neighbourhood of $\lambda_*$ contained in $B_\delta(\lambda_*)$. This follows from the same roughness argument, with coefficient perturbations of size $\mathcal{O}(\mathrm{e}^{-\alpha_0\min\{N_1,N_2,M_1\}})$ for some $\alpha_0 > 0$.}

The associated projections are again given by $P_{2\ell,k}^{i,\pm,j}(n;\lambda) = \Phi_{2\ell,k}^{i,\pm,j}(n,n;\lambda)$ where $k = 1,2$, $i = f,b$, and $j = s,u$. Moreover, the roughness theorem for exponential dichotomies gives that 
\begin{equation}\label{2PulseProofConvergingProjections}
    \begin{split} 
        |P_{2\ell,1}^{f,+,u}(N_1;\lambda) - P_*^u(0;\lambda)| &\leq C\mathrm{e}^{-\alpha N_1}, \\
        |P_{2\ell,1}^{b,-,s}(-N_1;\lambda) - P_*^s(0;\lambda)| &\leq C\mathrm{e}^{-\alpha N_1}, \\
        |P_{2\ell,1}^{f,-,s}(\rev{-M_1};\lambda) - P_0^s(0;\lambda)| &\leq C\mathrm{e}^{-\alpha M_1}, \\
        |P_{2\ell,2}^{b,+,u}(\rev{M_1};\lambda) - P_0^u(0;\lambda)| &\leq C\mathrm{e}^{-\alpha M_1},\\
        |P_{2\ell,2}^{f,+,u}(N_2;\lambda) - P_*^u(0;\lambda)| &\leq C\mathrm{e}^{-\alpha N_2}, \\
        |P_{2\ell,2}^{b,-,s}(-N_2;\lambda) - P_*^s(0;\lambda)| &\leq C\mathrm{e}^{-\alpha N_2}. 
    \end{split}
\end{equation}
Now, for $k = 1,2$ let $a_k = (a_k^{f,+},a_k^{b,-}) \in W_a$, $b = (b_1^{f,-},b_1^{f,+},b_1^{b,-},b_1^{b,+},b_2^{f,-},b_2^{f,+},b_2^{b,-},b_2^{b,+}) \in W_{b}$, $c = (c^{b,+},c^{f,-}) \in W_c$, and $\lambda \in W_\lambda$, where the spaces $W_a,W_{b,1},W_{b,2},W_c$ and $W_\lambda$ are defined as follows:
\begin{equation}
    \begin{split}
        W_a &:= \mathrm{Rg}(P_*^u(0;\lambda_*))\oplus\mathrm{Rg}(P_*^s(0;\lambda_*)) \\
        W_{b} &:= \bigoplus_{k=1}^2 \bigg(\mathrm{Rg}(P_{2\ell,k}^{f,-,u}(0;\lambda_*))\oplus\mathrm{Rg}(P_{2\ell,k}^{f,+,s}(0;\lambda_*))\bigg)\oplus\bigg(\mathrm{Rg}(P_{2\ell,k}^{b,-,u}(0;\lambda_*))\oplus\mathrm{Rg}(P_{2\ell,k}^{b,+,s}(0;\lambda_*))\bigg) \\
        W_c &:= \mathrm{Rg}(P_0^u(0;\lambda_*))\oplus\mathrm{Rg}(P_0^s(0;\lambda_*)) \\
        W_\lambda &:= B_\delta(\lambda_*) \subset \Omega,    
    \end{split}
\end{equation}
and each space is endowed with the maximum norm over each of its components. 

For $\delta > 0$ taken sufficiently small and $N_1,N_2,M_1 \geq 1$ taken sufficiently large, we can write solution to the eigenvalue problem \eqref{2PulseProof1} for the localized solution $\{u_n^{2\ell}(N_1,N_2,M_1)\}$ as
\begin{equation}\label{2PulseSolutionProfile}
    \begin{split}
        v_n^{f,1,-} &= \Phi_{2\ell,1}^{f,-,u}(n,0;\lambda)b_1^{f,-}, \quad n \leq 0 \\
        v_n^{f,1,+} &= \Phi_{2\ell,1}^{f,+,s}(n,0;\lambda)b_1^{f,+} + \Phi_{2\ell,1}^{f,+,u}(n,N_1;\lambda)a_1^{f,+}, \quad 0 \leq n \leq N_1 \\
        v_n^{b,1,-} &= \Phi_{2\ell,1}^{b,-,s}(n,-N_1;\lambda)a_1^{b,-} + \Phi_{2\ell,1}^{b,-,u}(n,0;\lambda)b_1^{b,-}, \quad -N_1 \leq n \leq 0 \\
        v_n^{b,1,+} &= \Phi_{2\ell,1}^{b,+,s}(n,0;\lambda)b_1^{b,+} + \Phi_{2\ell,1}^{b,+,u}(n,M_1;\lambda)c^{b,+} \quad  0 \leq n \leq M_1 \\   
        v_n^{f,2,-} &= \Phi_{2\ell,2}^{f,-,s}(n,-M_1;\lambda)c^{f,-} + \Phi_{2\ell,2}^{f,-,u}(n,0;\lambda)b_2^{f,-} \quad  -M_1 \leq n \leq 0 \\
        v_n^{f,2,+} &= \Phi_{2\ell,2}^{f,+,s}(n,0;\lambda)b_2^{f,+} + \Phi_{2\ell,\rev{2}}^{f,+,u}(n,N_2;\lambda)a_2^{f,+}, \quad 0 \leq n \leq N_2 \\
        v_n^{b,2,-} &= \Phi_{2\ell,2}^{b,-,s}(n,\rev{-N_2};\lambda)a_2^{b,-} + \Phi_{2\ell,2}^{b,-,u}(n,0;\lambda)b_2^{b,-}, \quad -N_2 \leq n \leq 0 \\
        v_n^{b,2,+} &= \Phi_{2\ell,2}^{b,+,s}(n,0;\lambda)b_2^{b,+}, \quad n \geq 0.
    \end{split}
\end{equation}
With the above form for the solution, we present the following lemma that solves the matching conditions \eqref{2PulseMatch1}, \eqref{2PulseMatch2}, and \eqref{2PulseMatch3}. 

\begin{lem}\label{lem:2PulseMatch}
    There exists $N_*,M_* \geq 1$ such that for all $N_1,N_2 \geq N_*$ and $M_1 \geq M_*$ the following holds uniformly in $(N_1,N_2,M_1)$. There exists operators $G_1:W_\lambda \times W_b \to W_a$, $G_2:W_\lambda \times W_b \to W_a$, and $G_3:W_\lambda \times W_b \to W_c$ with $a_1 = G_1(\lambda)b$ solves \eqref{2PulseMatch1}, $a_2 = G_2(\lambda)b$ solves \eqref{2PulseMatch2}, and $c = G_3(\lambda)b$ solves \eqref{2PulseMatch3}. The functions $G_1,G_2,$ and $G_3$ are analytic in $\lambda$ and linear in $b$, and satisfy
    \begin{equation}
        |G_1(\lambda)b| \leq C\mathrm{e}^{-\alpha N}|b|, \quad |G_2(\lambda)b| \leq C\mathrm{e}^{-\alpha N}|b|, \quad |G_3(\lambda)b| \leq C\mathrm{e}^{-\alpha N}|b|,
    \end{equation}
    where $N = \min\{N_1,N_2,M_1\}$.
\end{lem}

\begin{proof}
    The existence of the functions $G_1$, $G_2$, and $G_3$ with the properties given in the lemma is identical to the proof of Lemma~\ref{lem:1PulseMatch1} applied to each matching condition \eqref{2PulseMatch1}, \eqref{2PulseMatch2}, and \eqref{2PulseMatch3}, respectively, with the solution form \eqref{2PulseSolutionProfile} in. For brevity, we omit the details.  
\end{proof}

Lemma~\ref{lem:2PulseMatch} provides a solution to the matching equations \eqref{2PulseMatch1}-\eqref{2PulseMatch3} for any $b \in W_b$. Therefore, it remains to solve the last four matching equations \eqref{2PulseMatch4}-\eqref{2PulseMatch7}. From here the proof follows similarly to the proof of Theorem~\ref{thm:SinglePulse} in the previous subsection, and so we only highlight the key points and differences. 

First, replacing all $a_1,a_2\in W_a$ and $c \in W_c$ with the functions of $\lambda$ and $b$ in Lemma~\ref{lem:2PulseMatch} in the matching equations \eqref{2PulseMatch4}-\eqref{2PulseMatch7} results in the linear system 
\begin{equation}
\begin{split}
    &\begin{bmatrix}
        0 \\ 0
    \end{bmatrix} = \Bigg(\begin{bmatrix}
        P_{N_1,N_2,M_1,1} & 0 \\
        0 & P_{N_1,N_2,M_1,2}
    \end{bmatrix} \\ 
    & + \begin{bmatrix}
       \Phi_{2\ell,1}^{\rev{f},+,u}(0,N_1;\lambda)G_1(\lambda) & 0 & 0 & 0 \\ 0 & -\Phi_{2\ell,1}^{b,-,s}(0,-N_1;\lambda)G_1(\lambda) & 0 & 0 \\
       0 & 0 &\Phi_{2\ell,2}^{\rev{f},-,s}(0,N_2;\lambda)G_2(\lambda) & 0\\
       0 & 0 & 0 & -\Phi_{2\ell,2}^{b,-,s}(0,-N_2;\lambda)G_2(\lambda) &
    \end{bmatrix}\\ 
    &\qquad + \begin{bmatrix}
       0 & 0 & 0 & 0 \\ 0 & \Phi_{2\ell,1}^{b,+,u}(0,M_1;\lambda)G_3(\lambda) & 0 & 0 \\
       0 & 0 &\Phi_{2\ell,2}^{f,-,s}(0,-M_1;\lambda)G_3(\lambda) & 0\\
       0 & 0 & 0 & 0
    \end{bmatrix}\Bigg)b,
\end{split}
\end{equation}
where 
\begin{equation}
    P_{N_1,N_2,M_1,k}(\lambda) = \begin{bmatrix}
        - P_{2\ell,k}^{\rev{f},-,\rev{u}}(0;\lambda) & P_{2\ell,k}^{f,+,s}(0;\lambda) & 0 & 0 \\
        0 & 0 & - P_{2\ell,k}^{b,-,\rev{u}}(0;\lambda) & P_{2\ell,k}^{b,+,s}(0;\lambda)
    \end{bmatrix}
\end{equation}
for $k = 1,2$. We write this linear system for the unknowns $b \in W_b$ compactly as 
\begin{equation}
    0 = [\mathrm{diag}(P_{N_1,N_2,M_1,1}(\lambda),P_{N_1,N_2,M_1,2}(\lambda)) + R_{N_1,N_2,M_1}(\lambda)]b.
\end{equation}
\rev{We define $\hat P_{N_1,N_2,M_1,k}$ analogously to $\hat P_N$ in \eqref{Phat}, and first introduce the unnormalized matching determinant
\[
    \widetilde D_{2\ell,N_1,N_2,M_1}(\lambda)
    :=
    \det\!\left(
    \operatorname{diag}\!\left(
    \hat P_{N_1,N_2,M_1,1}(\lambda),
    \hat P_{N_1,N_2,M_1,2}(\lambda)
    \right)
    +\hat R_{N_1,N_2,M_1}(\lambda)
    \right).
\]
Here $\hat R_{N_1,N_2,M_1}$ denotes the restriction of the remainder to $W_b$. The matching problem has a nontrivial bounded solution if and only if this determinant vanishes.

For each interface, choose compatible analytic frames by setting
\[
     Z_{2\ell,k}^{i,\pm,j}(\lambda)
     :=
     P_{2\ell,k}^{i,\pm,j}(0;\lambda)Z_i^{\pm,j}(\lambda),
     \qquad k=1,2,
\]
for $(i,\pm,j)=(f,-,u),(f,+,s),(b,-,u),(b,+,s)$. Writing $L=\min\{N_1,N_2,M_1\}$, the projection estimates give
\[
     Z_{2\ell,k}^{i,\pm,j}(\lambda)-Z_i^{\pm,j}(\lambda)
     =\mathcal O(\mathrm e^{-\eta L})
\]
for some $\eta>0$, uniformly on a fixed compact neighbourhood of $\lambda_*$. These matrices therefore form frames for all sufficiently large $L$.

As in the single-pulse proof, choose the fixed bases of $W_b$ from these frames at $\lambda_*$. After shrinking the spectral neighbourhood if necessary, the resulting coordinate changes are analytic and uniformly bounded, together with their inverses, for all sufficiently large $L$. Their combined determinant $h_{N_1,N_2,M_1}(\lambda)$ is analytic and nonvanishing, with both $h_{N_1,N_2,M_1}$ and its inverse uniformly bounded. Define the normalized Evans function by
\[
     D_{2\ell,N_1,N_2,M_1}^U(\lambda)
     :=
     h_{N_1,N_2,M_1}(\lambda)
     \widetilde D_{2\ell,N_1,N_2,M_1}(\lambda).
\]
This normalization preserves zeros and their multiplicities. The compatible-frame estimates and the exponentially small matching remainder then yield
\[
     D_{2\ell,N_1,N_2,M_1}^U(\lambda)
     =
     [D_f^U(\lambda)D_b^U(\lambda)]^2
     +\mathcal O(\mathrm e^{-\gamma L}),
     \qquad \gamma>0,
\]
uniformly on the chosen compact neighbourhood, where $D_f^U$ and $D_b^U$ are defined in \eqref{EvansFrontBack}.} If we have the expansions \eqref{DfDbexpansions2Pulse} for $\rev{D_f^U}$ and $\rev{D_b^U}$ in a neighbourhood of $\lambda_*$, then  
\begin{equation}
     \rev{D_{2\ell,N_1,N_2,M_1}^U}(\lambda) = c_{2\ell}(\lambda - \lambda_*)^{2m_f + 2m_b} + \mathcal{O}(|\lambda - \lambda_*|^{2m_f + 2m_b + 1}) + \mathcal{O}(\mathrm{e}^{-\alpha \min\{N_1,N_2,M_1\}}),
\end{equation}
for a $c_{2\ell} \neq 0$. An application of Rouch\'e's theorem will then show that $D_{2\ell,N_1,N_2,M_1}$ has $2(m_f + m_b)$ roots in an $(N_1,N_2,M_1)$-independent neighbourhood of $\lambda_*$ for all $N_1,N_2,M_1$ taken sufficiently large. The argument is identical to that in the proof of Theorem~\ref{thm:SinglePulse} and \cite[Lemma~6.10]{makrides2019existence}. We omit the details.

\rev{A similar result to Corollary~\ref{cor:SinglePulse uniform} holds for 2 pulses.}

\begin{cor} \label{cor:2Pulse uniform}
    \rev{Assume Hypotheses~\ref{hyp:FixedPoints}, \ref{hyp:FrontBack}, \ref{hyp:Morse}, and~\ref{hyp:2Localized}. Let $G$ be a bounded domain with piecewise smooth boundary and $\overline G\subset U$. Suppose that neither $D_f^U$ nor $D_b^U$ vanishes on $\partial G$. Then, whenever $\min\{N_1,N_2,M_1\}$ is sufficiently large, the number of eigenvalues of the two-pulse state in $G$, counted with algebraic multiplicity, equals twice the sum of the corresponding counts for the front and back.}
\end{cor}

\begin{proof}
    \rev{Apply the finite-cover argument used for Corollary~\ref{cor:SinglePulse uniform}, now using Theorem~\ref{thm:2Pulse}. Each constituent eigenvalue contributes twice its algebraic multiplicity, and a finite subcover of the remaining compact set excludes additional eigenvalues whenever $\min\{N_1,N_2,M_1\}$ is sufficiently large.}
\end{proof}

%%%%%%%%%%%%%%%%%%%%%%%%%%%%%%%%%%%%%%%%%%%%%%%%%%%%%%%%%%%%%%%%%%%%%%%%%%%%%%%%%%%%%%%%%%%%%%
\section{Application to lattice dynamical systems}\label{sec:Application}

In this section, we apply our results to Nagumo-type lattices. \rev{Theorems~\ref{thm:SinglePulse} and~\ref{thm:2Pulse}, together with Corollaries~\ref{cor:SinglePulse uniform} and~\ref{cor:2Pulse uniform}, provide spectral approximation on compact subsets of $\Omega$. They do not describe eigenvalues approaching $\Sigma^0\cup\Sigma^*$ or exclude spectrum at arbitrarily large $|\lambda|$. For the scalar lattices considered here, however, Hypotheses~\ref{hyp:Localized} and~\ref{hyp:2Localized} uniformly bound the profiles and hence the norms of their linearized operators. Self-adjointness then confines their spectra to a fixed bounded real interval, independently of the separation parameters. If both asymptotic states are strictly spectrally stable, the nonnegative part of this interval lies inside $\Omega$, so the corollaries count all potentially unstable eigenvalues for sufficiently large separations. In particular, strictly spectrally stable fronts and backs generate strictly spectrally stable localized states. No separate near-zero interaction analysis is needed in this case, although zero eigenvalues may still occur at folds or other degeneracies.}

\rev{To demonstrate these results, we study a class of Nagumo-type lattices.} In particular, for $M \geq 1$ we consider the real-valued cubic–quintic Ginzburg–Landau equation on a rectangular lattice,
\begin{equation}\label{eq:rcqGL}
    \dot U_{n,m} = \theta (U_{n+1,m} + U_{n-1,m} + U_{n,m+1} + U_{n,m-1} - 4U_{n,m}) - \mu U_{n,m}+2U_{n,m}^3-U_{n,m}^5,
\end{equation}
for $n \in \mathbb{Z}$ and $m = 1,\dots,M$. Here $\theta$ denotes the coupling strength between neighboring lattice sites, and $\mu$ is a bifurcation parameter. We impose Neumann-type boundary conditions $U_{n,0} = U_{n,1}$ and $U_{n,M+1} = U_{n,M}$. In the case $M = 1$, the system reduces to \eqref{LDS} from the introduction, and we suppress the second index, writing $U_n$ in place of $U_{n,1}$.

In what follows, we focus primarily on the case $M = 1$, for which \eqref{eq:rcqGL} coincides with \eqref{LDS}. As illustrated in Figure~\ref{fig:IntroFig}, this system supports a rich family of localized structures whose stability can be analyzed using our results. In the final subsection, we briefly consider the case $M > 1$, emphasizing computational results while noting that many of the analytical conclusions for $M = 1$ extend in a straightforward manner to higher dimensions.

%%%%%%%%%%%%%%%%%%%%%%%%%%%%%%%%%%%%%%%%%%%%%%%%%%%%%%%%%%%%%%%%%%%%%%%%%%%%%%%%%%%%%%%%%%%%%%
\subsection{Basic properties of the model}

We begin by reviewing key properties of the model \eqref{eq:rcqGL} and recasting it to align with the framework of our theoretical results. Lattice systems of the form \eqref{eq:rcqGL} have been studied previously in \cite{bramburger2020spatially,taylor2010lattice}, primarily in the case $M = 1$. Here, we summarize those results and extend them to $M > 1$ where appropriate.

For $\mu \in [0,1]$, the model \eqref{eq:rcqGL} admits five spatially homogeneous steady states given by $U_{n,m} = 0$ and $U_{n,m} = \pm U_\pm(\mu)$, where
\begin{equation}
    U_\pm(\mu):=  \sqrt{1\pm\sqrt{1-\mu}}       
\end{equation}
The nontrivial branches $\pm U_-(\mu)$ emerge via a subcritical pitchfork bifurcation from the trivial state at $\mu = 0$. At $\mu = 1$, two saddle-node bifurcations occur, where $\pm U_+(\mu)$ and $\pm U_-(\mu)$ coalesce. In the absence of spatial coupling ($\theta = 0$), the trivial state and $\pm U_+(\mu)$ are stable, while $\pm U_-(\mu)$ are unstable, yielding the bistable structure underlying localized pattern formation.

The system \eqref{eq:rcqGL} admits a gradient flow structure on $\ell^2$, which can be written as $\dot{U}_{n,m} = -\partial \mathcal{E}/\partial U_{n,m}$, where the energy functional $\mathcal{E} : \ell^2 \to \mathbb{R}$ is given by
\begin{equation}
    \mathcal{E}(\{U_{n,m}\}) = \sum_{n\in\mathbb{Z}}\sum_{m=1}^M\left(\frac{\theta}{2}(U_{n+1,m}-U_{n,m})^2 + \frac{\theta}{2}(U_{n,m+1}-U_{n,m})^2 +\frac{1}{2}\mu U_{n,m}^2 - \frac{1}{2} U_{n,m}^4 + \frac{1}{6}U_{n,m}^6\right)
\end{equation}
In particular, $\dot{\mathcal{E}} \leq 0$, so solutions with initial data in $\ell^2$ evolve toward equilibrium as $t \to \infty$. For spatially homogeneous steady states $U_{n,m} = U^*$, the energy per lattice site is
\begin{equation}
    \mathcal{E}(U^*)=\frac{1}{2}\mu (U^*)^2 - \frac{1}{2} (U^*)^4 + \frac{1}{6}(U^*)^6.    
\end{equation} 
The trivial state $U^* = 0$ has zero energy, while the energy of the upper state $U^* = U_+(\mu)$ depends on $\mu$. At $\mu = 0.75$, these energies coincide; we refer to this parameter value as the Maxwell point. This point serves as an organizing center for localized pattern formation, marking the parameter regime where competing stable states are energetically balanced; see \cite{bramburger2024localized} for further discussion.

\begin{figure}[t] %Figure: Snakes and Ladders diagram for LDS
    \center
    \includegraphics[width = 0.99\textwidth]{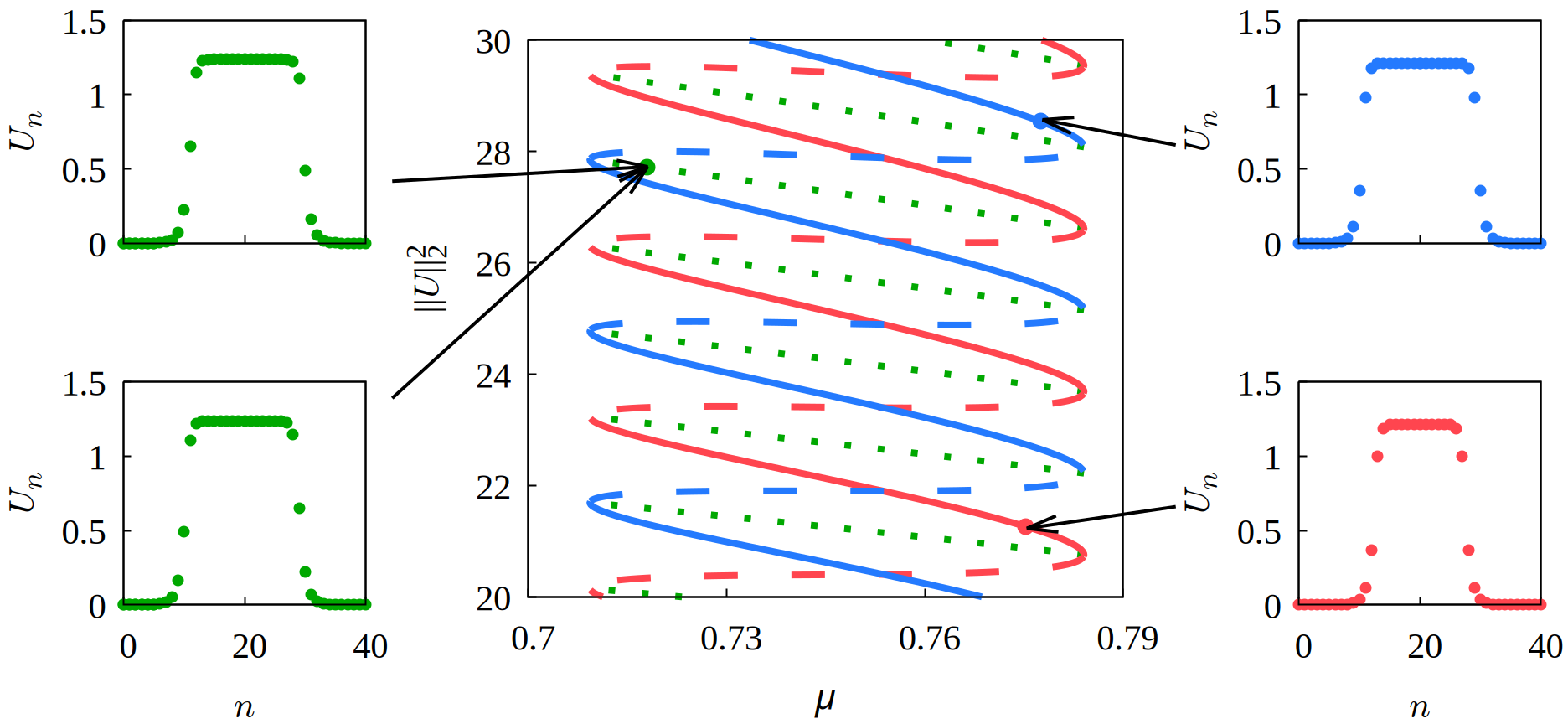}  
    \caption{Bifurcation diagram of single-region localized patterns for \eqref{eq:rcqGL} with $M=1$ at $\theta=0.5$. Symmetric on-site solutions are red, symmetric off-site are blue, and asymmetric solutions connecting the saddle-nodes are green. Solid lines indicate stability; dashed/dotted lines indicate instability. Insets show representative solution profiles, with on-site plateaus having an odd number of points and off-site plateaus having an even number.}
    \label{fig:snakes ladders}
\end{figure}

As described in the introduction, the steady-state equation $\dot U_{n,m} = 0$ can be recast as a discrete dynamical system in $\mathbb{R}^{2M}$,
\begin{subequations}\label{CGLmap}
    \begin{align}
    u_{n+1,m}^{(1)} &= u_{n+1,m}^{(2)}\\
    u_{n+1,m}^{(2)} &= 4u_{n,m}^{(2)} - u_{n,m}^{(1)} - u_{n,m+1}^{(2)} - u_{n,m-1}^{(2)} +\frac{1}{\theta}\left(\mu u_{n,m}^{(2)}-2(u_{n,m}^{(2)})^3+(u_{n,m}^{(2)})^5\right).
\end{align}
\end{subequations}

for $m = 1,\dots,M$. The works \cite{bramburger2020spatially,bramburger2021isolas} rigorously establish the existence and bifurcation structure of localized solutions in the case $M = 1$. In particular, they show that symmetric and asymmetric localized states with a single region of localization, as in Figures~\ref{fig:IntroFig}(c) and \ref{fig:IntroFig}(d), give rise to the classical snakes-and-ladders bifurcation structure illustrated in Figure~\ref{fig:snakes ladders}. In contrast, solutions with multiple regions of localization or oscillatory tails, as in Figures~\ref{fig:IntroFig}(e) and \ref{fig:IntroFig}(f), lie on isolas, forming closed, figure-eight-shaped bifurcation curves. However, the stability of these solutions has not yet been analyzed.

The map \eqref{CGLmap} exhibits a reversible structure that can be exploited to understand the existence of front and back solutions, as well as to classify localized steady-states as symmetric or asymmetric. The map acts on vectors 
\begin{equation}
   (u^{(1)},u^{(2)}):= (u_{1}^{(1)},\dots,u_M^{(1)},u_{1}^{(2)},\dots,u_M^{(2)}) \in \mathbb{R}^{2M}
\end{equation}
and the reverser $\mathcal{R}:\mathbb{R}^{2M} \to \mathbb{R}^{2M}$ is given by the involution
\begin{equation}
    \mathcal{R}(u^{(1)},u^{(2)}) = (u^{(2)},u^{(1)}).
\end{equation}
This structure implies that if $\{(u_n^{(1)},u_n^{(2)})\}_{n \in \mathbb{Z}}$ is a trajectory of \eqref{CGLmap}, then so is $\{\mathcal{R}(u_{-n}^{(1)},u_{-n}^{(2)})\}_{n \in \mathbb{Z}}$. A solution is called {\bf symmetric} if it is invariant under this reverser, i.e.,
\begin{equation}
    \mathcal{R}\{(u_n^{(1)},u_n^{(2)})\}_{n \in \mathbb{Z}} = \{(u_n^{(1)},u_n^{(2)})\}_{n \in \mathbb{Z}}.    
\end{equation}
Among symmetric solutions, we define a solution as {\bf on-site} if there exists an $n \in \mathbb{Z}$ so that $\mathcal{R}(u_{n}^{(1)},u_{n}^{(2)}) = (u_{n}^{(1)},u_{n}^{(2)})$, and {\bf off-site} otherwise. Examples of symmetric on-site and off-site localized solutions, as well as asymmetric solutions, are shown in Figure~\ref{fig:snakes ladders}. In the context of our theoretical results in Section~\ref{sec:Results}, we will have $d = 2M$, and the reversible structure ensures that the asymptotic Morse indices $i_\infty$ are constrained to be $M$ which readily verifies Hypothesis~\ref{hyp:Morse} and provides the basis for applying our Evans function analysis to the stability of localized solutions.   

\begin{figure}[t]
    \center
    \includegraphics[width=0.99\textwidth]{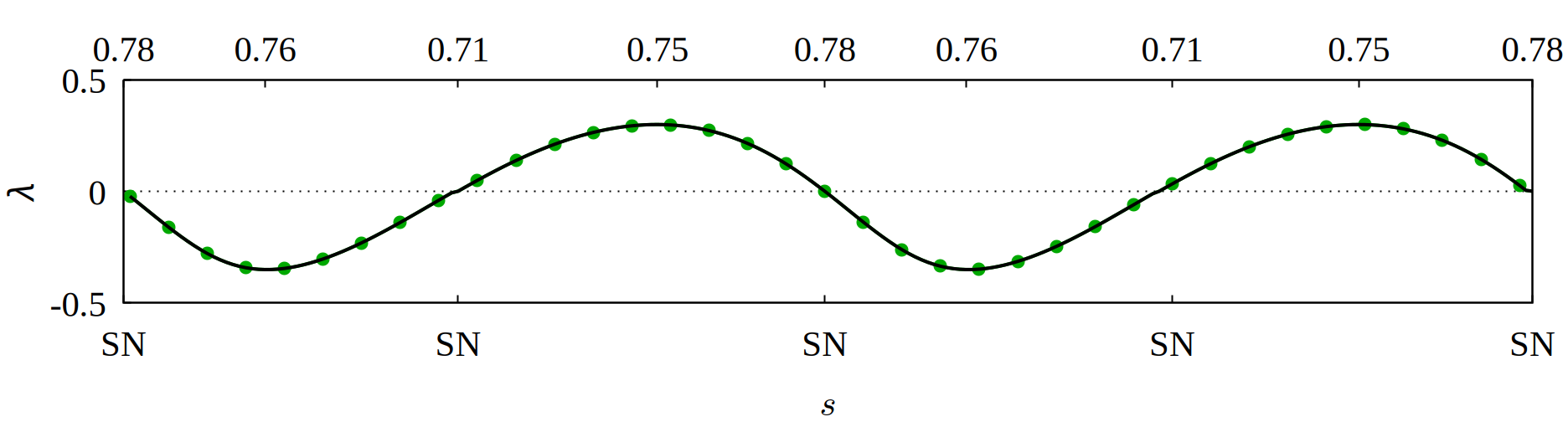}
    \caption{Largest real-part eigenvalues computed numerically as a function of arclength $s$ over two full cycles of the snaking region (red and blue branches) in Figure~\ref{fig:snakes ladders}. There are two identical leading eigenvalues, one inherited from the front (solid line) and the other from the back (dots). The x-axis labels `SN' indicate the locations of the saddle-nodes and the upper horizontal axis shows a non-linear scale of the bifurcation parameter $\mu$ as a function of arclength.} \label{fig:snakes eigs}
\end{figure}

\begin{figure}[t]
    \center
    \includegraphics[width=0.99\textwidth]{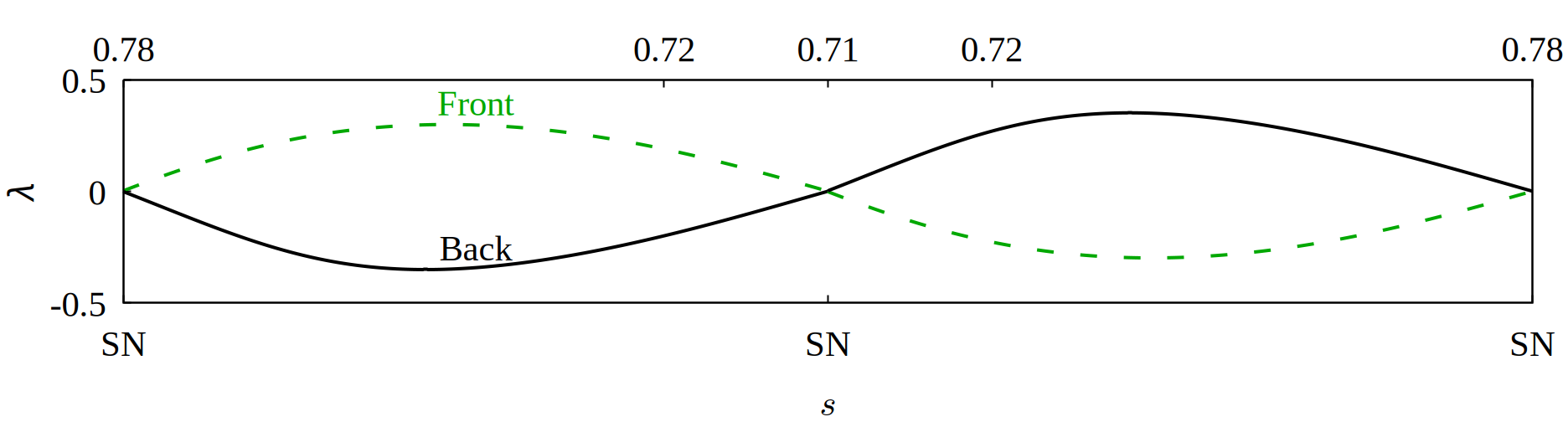}
    \caption{Largest real-part eigenvalues computed numerically as a function of arclength $s$ over the `ladder' states (green branches) in Figure~\ref{fig:snakes ladders}. The x-axis labels `SN' indicate the locations of the saddle-nodes and the upper horizontal axis shows a non-linear scale of the bifurcation parameter $\mu$ as a function of arclength. The value $\mu=0.72$ denotes the location of the asymmetric states shown in Figure~\ref{fig:snakes ladders}. The top-left panel of Figure~\ref{fig:snakes ladders} corresponds to the first half and the bottom-left to the second half of this figure. } \label{fig:rungs eigs}
\end{figure}

%%%%%%%%%%%%%%%%%%%%%%%%%%%%%%%%%%%%%%%%%%%%%%%%%%%%%%%%%%%%%%%%%%%%%%%%%%%%%%%%%%%%%%%%%%%%%%
\subsection{Stability of flat plateau localized solutions with $M = 1$}

In this section, we fix $M = 1$ in \eqref{eq:rcqGL} and apply our stability results to localized solutions with flat plateaus, i.e., those appearing in the snaking bifurcation diagram of Figure~\ref{fig:snakes ladders}. The existence of steady-state front and back solutions for \eqref{eq:rcqGL} was established in \cite{bramburger2020spatially} by continuing explicit solutions from the anti-continuum limit $\theta = 0$ into $\theta > 0$. In the anti-continuum limit, the front solutions take the form $\overline{U}^0(\mu)=\{\overline{u}_n(\mu)\}_{n\in\mathbb{Z}}$ with
\begin{align} \label{eq:front d0}
    \overline{u}_n(\mu)=\begin{cases}
        0, & n\leq0,\\
        U_+(\mu), & n>0,
    \end{cases}
\end{align}
and $\overline{U}^\pm(\mu)=\{\overline{v}^\pm_n(\mu)\}_{n\in\mathbb{Z}}$ with
\begin{align} \label{eq:front d0 inter}
    \overline{v}^\pm_n=\begin{cases}
        0, & n<0,\\
        \pm U_-(\mu), & n=0,\\
        U_+(\mu), & n>0.
    \end{cases}
\end{align}

Using similar continuation techniques, we can compute the eigenvalues of these solutions as $\theta\to 0$. The following proposition summarizes the existence and bifurcation results from \cite{bramburger2020spatially} while extending them to include stability calculations. We state the result without proof, as the stability computations follow in a manner analogous to those in \cite{bramburger2020localized}, albeit in a slightly more delicate setting. Because the front solutions are not localized, the analysis is carried out in $\ell^\infty$, the Banach space of uniformly bounded sequences equipped with the supremum norm.

\begin{prop}\label{prop:FlatStability}
    There exists a $\theta_* > 0$ so that the following is true:
    \begin{enumerate}
        \item \textbf{Persistence and stability of fronts for $\mu \in (0,1)$:}  Let $K\subset (0,1)$ be a compact interval. Choose a continuous function $U^\ast: K\to\ell^\infty$ defined either by \eqref{eq:front d0} or \eqref{eq:front d0 inter}. Then there exist $\delta>0$ and a smooth function $V^\ast:K\times (-\theta_\ast,\theta_\ast)\to\ell^\infty$ with $V^\ast(\mu,0)=U^\ast(\mu)$ for all $\mu\in K$, such that
    \begin{itemize}
        \item \textbf{Persistence:} For all $(\mu,\theta) \in K \times (-\theta_*,\theta_*)$, $V^*(\mu,\theta)$ is a steady-state solution of \eqref{eq:rcqGL} with $M = 1$.
        \item \textbf{Stability:} 
        \begin{itemize}
            \item If $V^*(\mu,0) = \overline{U}^\pm(\mu)$, then the linearization of \eqref{eq:rcqGL} about $V^*(\mu,\theta)$ has exactly one positive eigenvalue, and the remainder of the spectrum lies on the negative real line, for all $(\mu,\theta) \in K \times (-\theta_*,\theta_*)$.
            \item If $V^*(\mu,0) = \overline{U}^0(\mu)$, then the spectrum of the linearization about $V^*(\mu,\theta)$ is entirely contained in the negative real line for all $(\mu,\theta) \in K \times (-\theta_*,\theta_*)$.
        \end{itemize}
    \end{itemize}
    \item \textbf{Fold bifurcations near $\mu = 0$:}  
    There exists a constant $\mu_1>0$ and a smooth function $\mu_l:[0,\theta_*]\to[0,\mu_1]$ such that, for each $\theta \in (0,\theta_*]$, the following hold:  
    \begin{itemize}
    \item \textbf{Existence:} There is a pair of symmetric steady-state solutions $U_l(\mu,\theta)$ and $V_l(\mu,\theta)$ of \eqref{eq:rcqGL} with $M = 1$ that bifurcate at a fold at $\mu = \mu_l(\theta)$ and exist for all $\mu \in [\mu_l(\theta),\mu_1]$. These solutions are smooth in $(\mu,\theta)$, and for each fixed $\mu$, we have
    \[
        U_l(\mu,\theta) \to \overline{U}^0(\mu), \quad V_l(\mu,\theta) \to \overline{U}^+(\mu) \quad \text{as } \theta \searrow 0.
    \]  
    \item \textbf{Expansion:} The fold curve satisfies
    \[
        \mu_l(\theta) = \frac{3}{\sqrt[3]{2}}\, \theta^{2/3} + \mathcal{O}(\theta) \quad \text{as } \theta \to 0.
    \]  
    \item \textbf{Stability:} At the fold bifurcation, the spectrum of the linearization of \eqref{eq:rcqGL} is contained on the negative real line except for exactly one eigenvalue that crosses zero transversely. For $(\mu,\theta)\in (\mu_l(\theta),\mu_1] \times [0,\theta_*]$, $U_l(\mu,\theta)$ is linearly stable, while $V_l(\mu,\theta)$ is unstable.
    \end{itemize} 
    \item \textbf{Fold bifurcations near $\mu = 1$:} There exists a constant $\mu_2>0$ and a smooth function $\mu_r:[0,\theta_*] \to (\mu_2,1]$ such that, for each $\theta \in (0,\theta_*]$, the following hold:
    \begin{itemize}
        \item \textbf{Existence:} There is a pair of symmetric steady-state solutions $U_r(\mu,\theta)$ and $V_r(\mu,\theta)$ of \eqref{eq:rcqGL} with $M = 1$ that bifurcate at a fold at $\mu = \mu_r(\theta)$ and exist for all $\mu \in [\mu_2, \mu_r(\theta)]$. These solutions are smooth in $(\mu,\theta)$, and for each fixed $\mu$, we have
        \[
            U_r(\mu,\theta) \to \overline{U}^0(\mu), \quad V_r(\mu,\theta) \to S^{-1}\overline{U}^+(\mu) \quad \text{as } \theta \searrow 0,
        \]
        where $S:\ell^\infty \to \ell^\infty$ is the left shift operator defined by $[Su]_n := u_{n+1}$.
        \item \textbf{Expansion:} The fold curve satisfies
        \[
            \mu_r(\theta) = 1 - \theta + \mathcal{O}(\theta^{3/2}) \quad \text{as } \theta \to 0.
        \]
        \item \textbf{Stability:} At the fold bifurcation, the spectrum of the linearization of \eqref{eq:rcqGL} is contained on the negative real line except for exactly one eigenvalue that crosses zero transversely. For $(\mu,\theta)\in [\mu_2,\mu_r(\theta)) \times [0,\theta_*]$, $U_r(\mu,\theta)$ is linearly stable, while $V_r(\mu,\theta)$ is unstable.
    \end{itemize}
    \end{enumerate}
\end{prop}

In plain terms, Proposition 4.1 shows that the front and back solutions of the lattice system \eqref{eq:rcqGL} not only persist for small but positive coupling $\theta$, but also have a precisely characterized spectrum: fronts continued into $\theta>0$ from $\overline{U}^+(\mu)$ have exactly one unstable eigenvalue, while those continued from $\overline{U}^0(\mu)$ are stable. Although the results are stated explicitly for front solutions, the reversible structure of the map transforms fronts into corresponding back solutions without altering the spectrum. Consequently, the stability properties for fronts immediately carry over to backs. This verifies Hypothesis~\ref{hyp:FrontBack}, and together with the existence and construction of localized solutions with flat plateaus as in \cite{bramburger2020spatially}, Hypothesis~\ref{hyp:Localized} is also satisfied. Therefore, the framework developed in Section 2 applies directly to the study of localized steady-states in the snaking regime of \eqref{eq:rcqGL}.

\begin{figure}[t]
    \center
    \includegraphics[width=0.99\textwidth]{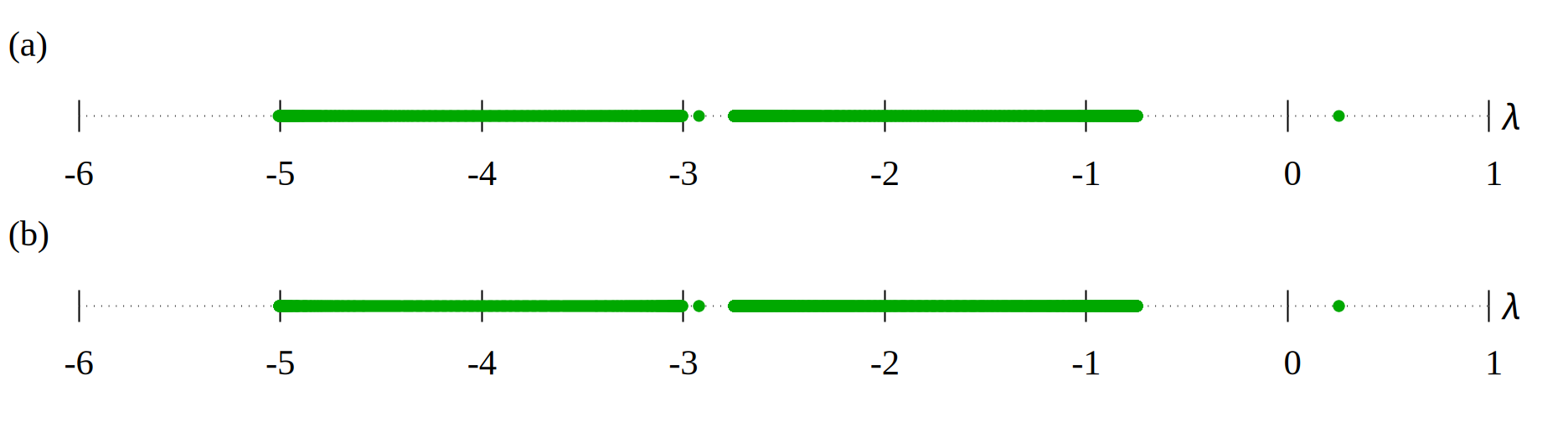}
    \caption{Numerically computed spectra for (a) a front solution and (b) a symmetric localized solution constructed from the front. The two solid regions show approximations of the essential spectrum (see Lemmas~\ref{lem:EssentialSpectrum} and \ref{lem:FrontBackEssential}) and the remaining two points denote elements of the point spectrum. Note that the spectra between the front and localized solutions are close and the point spectrum of the localized solutions has multiplicity 2, as predicted by Theorem~\ref{thm:SinglePulse}.}
    \label{fig:front_local_compare}
\end{figure}

We illustrate Theorems~\ref{thm:SinglePulse} and~\ref{thm:2Pulse} through numerical computations. Figure~\ref{fig:snakes ladders} shows part of the snaking bifurcation diagram, including stability information. Figures~\ref{fig:snakes eigs}, \ref{fig:rungs eigs}, and~\ref{fig:front_local_compare} compare the spectra of localized patterns with those of their constituent fronts and backs. \rev{The computed unstable symmetric single-plateau solutions have two positive eigenvalues, one associated with each interface, whereas the asymmetric `ladder' solutions have one positive eigenvalue.} These observations are consistent with Theorem~\ref{thm:SinglePulse}. Figure~\ref{fig:front_local_compare} shows the \rev{numerically computed spectra on a large bounded domain} for a front and a corresponding symmetric localized solution, illustrating the agreement between their isolated eigenvalues. Finally, Figure~\ref{fig:2pulse} shows symmetric two-pulse solutions with 0, 2, or 4 computed unstable eigenvalues, consistent with Theorem~\ref{thm:2Pulse}. \rev{These computations illustrate the asymptotic results but do not verify that the displayed solutions exceed the separation thresholds required by the theorems.}

\begin{figure}
    \centering
    \includegraphics[width=0.99\textwidth]{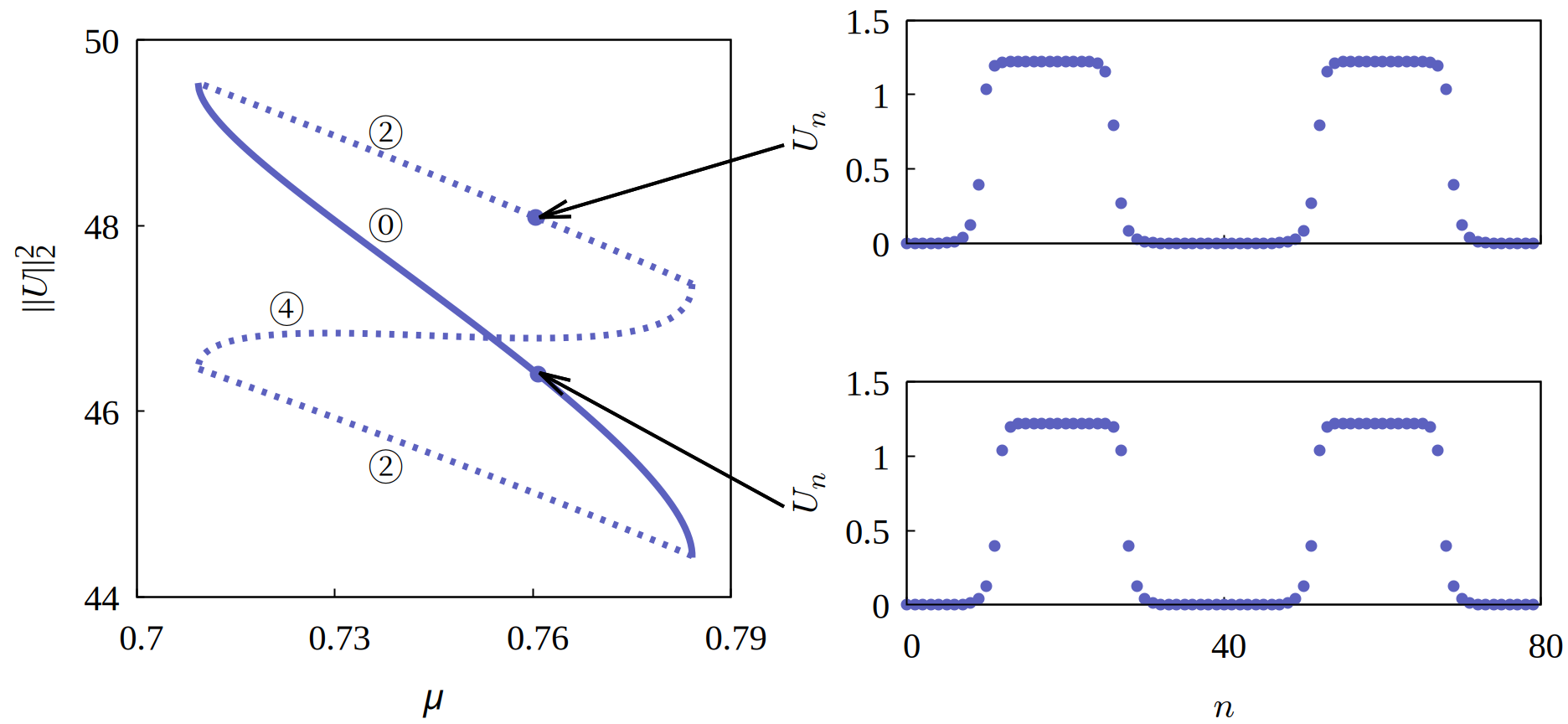}
    \caption{Bifurcation diagram of 2-pulse patterns for \eqref{eq:rcqGL} with $M=1$ at $\theta=0.5$. The width of the pulses differs by one; therefore, one pulse is on-site, and the other is off-site. Solid lines indicate stability; dashed/dotted lines indicate instability. The circled numbers indicate the number of unstable eigenvalues along each segment of the branch. Insets show representative solution profiles.}
    \label{fig:2pulse}
\end{figure}

%%%%%%%%%%%%%%%%%%%%%%%%%%%%%%%%%%%%%%%%%%%%%%%%%%%%%%%%%%%%%%%%%%%%%%%%%%%%%%%%%%%%%%%%%%%%%%
\subsection{Oscillatory plateau localized solutions with $M = 1$}

\begin{figure}
    \centering
    \includegraphics[width=0.99\textwidth]{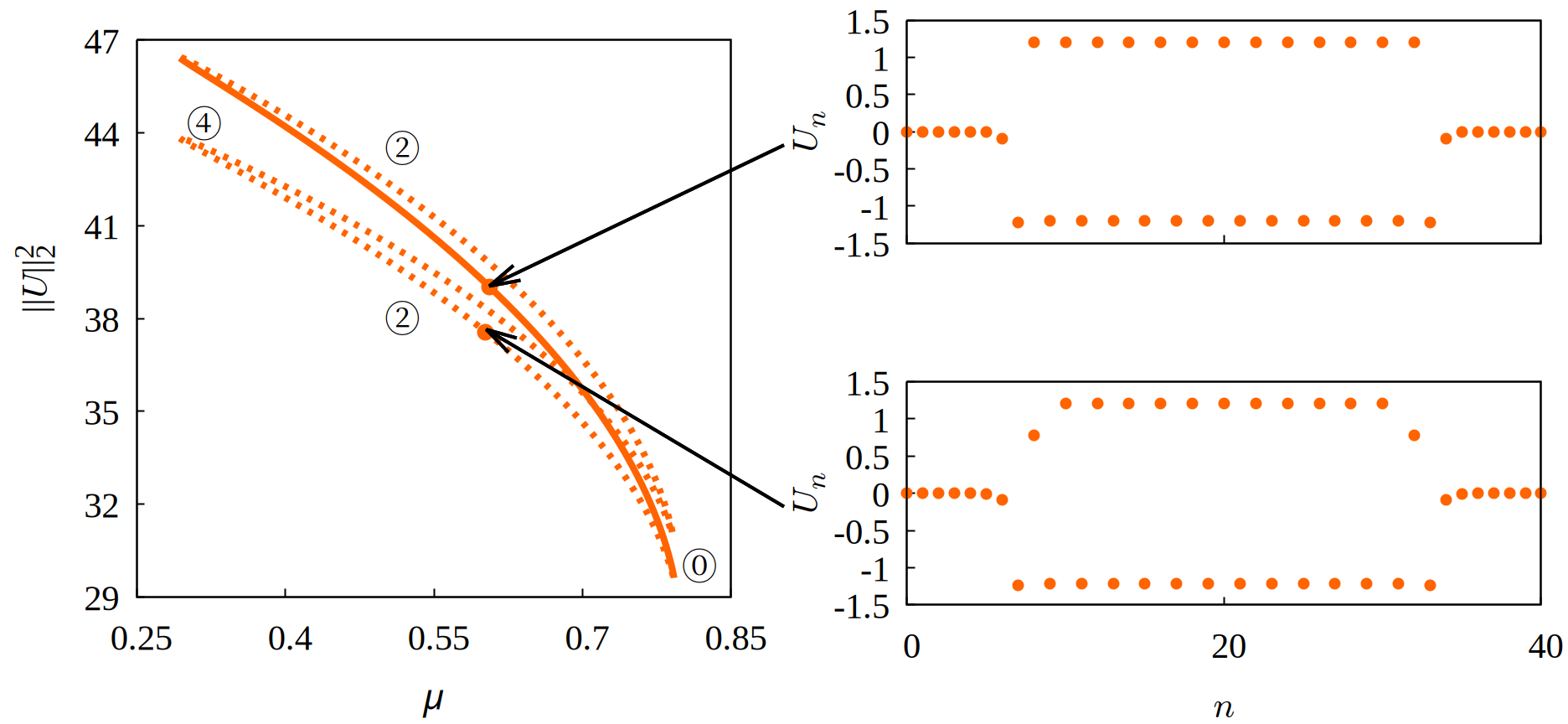}
    \caption{Bifurcation diagram of single-region oscillatory plateau localized patterns for \eqref{eq:rcqGL} with $M=1$ at $\theta=0.05$. Solid lines indicate stability; dashed/dotted lines indicate instability. The circled numbers indicate the number of unstable eigenvalues along each segment of the branch. Insets show representative solution profiles.}
    \label{fig:2cycle}
\end{figure}

We now turn our attention to localized solutions with oscillatory plateaus in the case $M = 1$. These patterns can be understood by considering fronts that connect the trivial fixed point to a $k$-cycle of the spatial map \eqref{CGLmap}. \rev{Sampling every $k$th lattice site selects a single point of this cycle, which is a fixed point of the $k$th iterate map. Under this map, the sampled front is a heteroclinic orbit from the trivial fixed point to the selected cycle point, while the phase-matched back provides the reverse connection. The full localized pulse is homoclinic to the trivial fixed point and spends a long stretch near the selected cycle point. As discussed in Section~\ref{sec:StabilityResult}, the associated $k$-step transfer matrix depends polynomially on $\lambda$ and therefore fits within the analytic framework of \eqref{StabilityMap}. Our results therefore apply to phase-matched fronts and backs whose sampled profiles satisfy the stated hypotheses, with separation measured in $k$-site blocks.}

Fronts (and, by reversibility, backs) asymptotic to a 2-cycle as $n \to \infty$ were proven to exist in \cite{bramburger2020spatially} by exploiting the `staggering' symmetry identified in \cite{taylor2010lattice}, which leaves steady-state solutions of \eqref{eq:rcqGL} invariant under
\begin{equation}
(\{U_n\}_{n \in \mathbb{Z}},\mu,\theta) \mapsto (\{(-1)^nU_n\}_{n \in \mathbb{Z}},\mu-4\theta,-\theta).
\end{equation}
Monotone fronts continued for $\theta<0$ from the anti-continuum solutions $\overline U^0(\mu)$ and $\overline U^\pm(\mu)$ can thus be mapped, via this staggering symmetry, to fronts for $\theta>0$ that asymptote to a 2-cycle. 

In \cite{bramburger2020spatially}, these localized solutions with a 2-cycle plateau were shown to form isolas in parameter space. While we omit a full analogue of Proposition~\ref{prop:FlatStability} for brevity, one can verify that fronts continued from $\overline U^\pm(\mu)$ under this procedure have exactly one positive eigenvalue and all remaining eigenvalues contained on the negative real line, while those continued from $\overline U^0(\mu)$ are fully stable. A fourth segment of the branch that completes the figure-eight-shaped isola \cite{bramburger2020spatially} is continued from 
\begin{align}
    \overline{w}_n(\mu)=\begin{cases}
        0, & n<0,\\
        -U_-(\mu), & n=0,\\
        U_-(\mu), & n=1, \\
        U_+(\mu), & n>0,
    \end{cases}
\end{align}
which has 2 unstable eigenvalues. As before, the reverser provides the corresponding back solutions from these fronts. Figure~\ref{fig:2cycle} shows the resulting bifurcation diagram of one such isola, demonstrating that there are either 0, 2, or 4 unstable eigenvalues resulting from the front solutions, as predicted.

\begin{figure}
    \centering
    \includegraphics[width=0.99\textwidth]{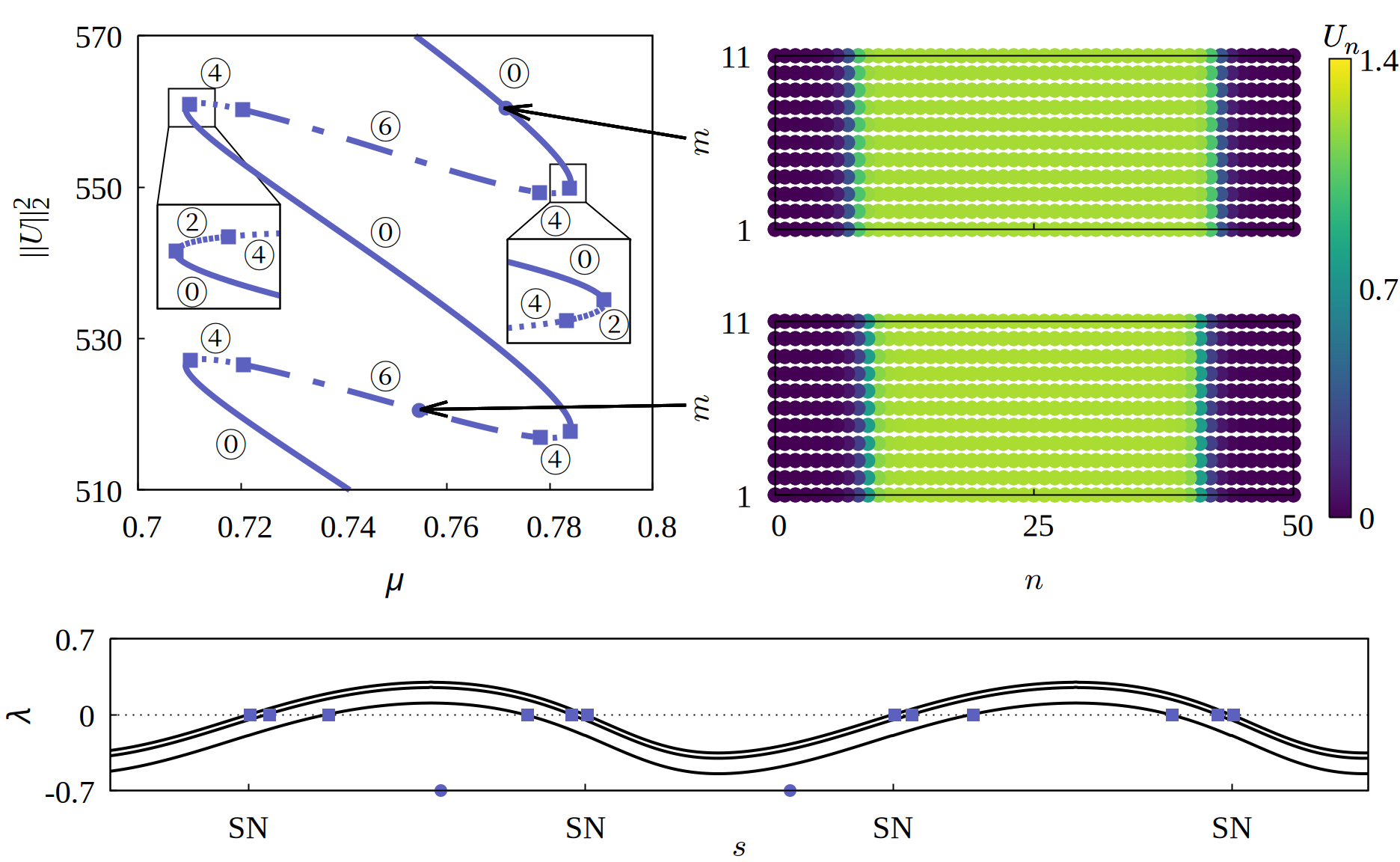}
    \caption{Bifurcation diagram of single stripe solutions for \eqref{eq:rcqGL} with $M=11$ at $\theta=0.5$. Solid lines indicate stability; dotted, dashed, and dot-dashed lines denote linearly unstable solutions with 2, 4, and 6 unstable eigenvalues, respectively. The circled numbers indicate the number of unstable eigenvalues along each segment of the branch, separated by squares. The insets in the left panel zoom into the folds, demonstrating the location of additional crossings of 2 eigenvalues. Representative solution profiles are shown on the right. The bottom panel shows the leading eigenvalues of the stripe solutions. Each curve represents two eigenvalues, one inherited from the front and the other from the back. The squares denote the locations along the branch where eigenvalues cross the imaginary axis. The circles along the x-axis correspond to the circles shown on the bifurcation diagram, where a solution is shown.}
    \label{fig:2D bif}
\end{figure}

%%%%%%%%%%%%%%%%%%%%%%%%%%%%%%%%%%%%%%%%%%%%%%%%%%%%%%%%%%%%%%%%%%%%%%%%%%%%%%%%%%%%%%%%%%%%%%

\subsection{Extension to rectangular lattices}

We now turn to the case $M > 1$ and consider extensions to rectangular lattices. Much of the analysis for $M = 1$ carries over: fronts can be identified along one spatial direction, and, by reversibility, the corresponding back solutions follow automatically. Moreover, our general framework for stability is sufficiently flexible to apply to localized solutions constructed from these fronts and backs, including multi-dimensional patterns with flat or oscillatory plateaus. In two dimensions, one can anticipate localized structures such as stripes, rectangles, or multi-pulse arrays. Here, we focus on a numerical demonstration of these phenomena, illustrating the existence and stability of such solutions, while refraining from a full—or even partial—analytical treatment.

Figures~\ref{fig:2D bif} and~\ref{fig:2D bif non} show two branches of localized solutions with $M=11$. Figure~\ref{fig:2D bif} shows symmetric on-site localized solutions that are uniform in $m$, i.e. stripe solutions. The branch of solutions snakes, as expected from its 1D counterpart (see Figure~\ref{fig:snakes ladders}), but on the rectangular lattice, we observe additional instabilities. As predicted, stripe solutions can exhibit 0, 2, 4, or 6 unstable eigenvalues, i.e., an even number. Figure~\ref{fig:2D bif non} shows symmetric spot solutions (panel (b)) and the corresponding fronts (panel (a)). The resulting branches are isolas. Panel (c) shows the number of unstable eigenvalues along each segment of the isola. The front solutions have 0, 2, or 4 unstable eigenvalues due to the reflective symmetry in $m$, and the spots have 0, 4, or 8. Both figures demonstrate consistency with the conclusions of Theorem~\ref{thm:SinglePulse}, suggesting that our results can be extended to rectangular lattices and likely to higher dimensions.

\begin{figure}
    \centering
    \includegraphics[width=0.99\textwidth]{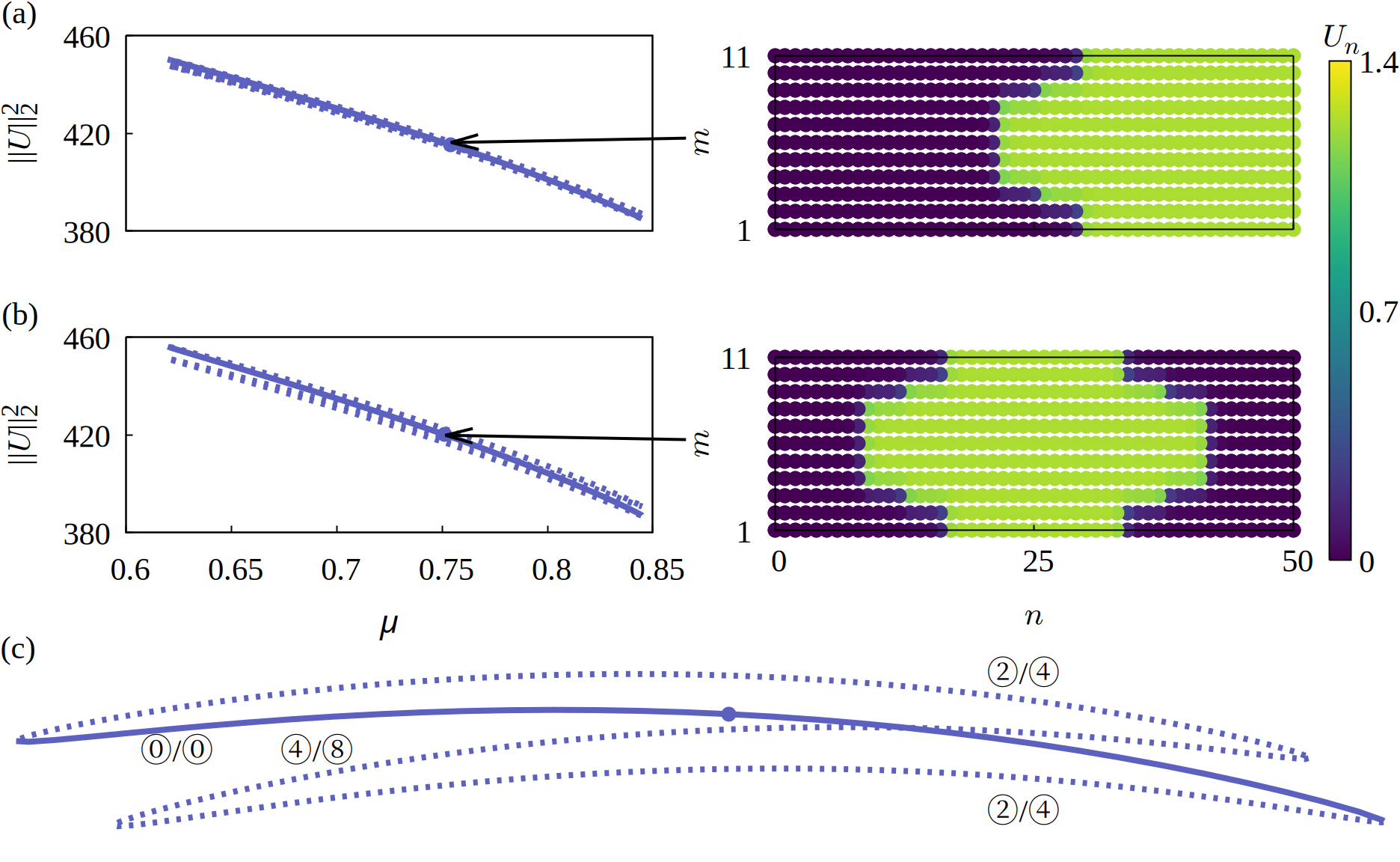}
    \caption{Bifurcation diagrams of 2D (a) fronts and (b) corresponding symmetric spot solutions with non-uniform interfaces for \eqref{eq:rcqGL} with $M=11$ at $\theta=0.1$. Solid lines indicate stability; dashed lines denote linearly unstable solutions. Representative solution profiles are shown on the right. Panel (c) shows a rotated view of the front/spot branch, where the circled numbers indicate the number of unstable eigenvalues along each segment of the branch for the fronts/spots, respectively. The circle on the branch denotes the location of the solutions shown on the right.}
    \label{fig:2D bif non}
\end{figure}

%%%%%%%%%%%%%%%%%%%%%%%%%%%%%%%%%%%%%%%%%
\section{Discussion} \label{sec:discussion}

In this work, we provide theoretical contributions that extend previous PDE-based studies by establishing a general framework for analyzing the stability of localized solutions in lattice dynamical systems via the properties of front and back solutions. Using a discrete Evans-function approach, we connect the spectral stability of a localized pattern to the spectra of its constituent fronts and backs, enabling a systematic characterization of eigenvalues for both single- and multi-pulse localized states. \rev{Our Evans-function factorization applies on compact subsets of the domain of consistent splitting and does not, by itself, provide spectral information near the essential spectrum. For the scalar Nagumo lattices, Lemma~\ref{lem:EssentialSpectrum} places the essential spectrum strictly in the left half-plane provided both asymptotic states are strictly spectrally stable.} Our framework is sufficiently general to accommodate one- and higher-dimensional lattices\rev{, as long as only one dimension is unbounded}, arbitrary numbers of disconnected regions of localization, and fronts asymptotic to $k$-cycles. This approach not only rigorously confirms the stability properties of many previously known localized solutions but also provides a practical tool for investigating new classes of patterns whose stability was previously inaccessible.

We illustrated the application of our results using the real-valued cubic–quintic Ginzburg–Landau equation on one- and two-dimensional lattices. In the one-dimensional case, we analyzed the stability of localized solutions with flat and oscillatory plateaus, confirming the predicted eigenvalue structure for single- and multi-pulse patterns. These findings naturally fall within the broader context of Nagumo-type lattice dynamical systems, which have long served as a prototypical setting for studying propagation failure and front-back interactions. Extending the framework to rectangular lattices ($M > 1$), we demonstrated numerically that fronts along a single spatial direction generate multi-dimensional localized patterns, with reversibility providing the corresponding back solutions, and that the Evans-function-based spectral characterization remains valid.

Looking ahead, an important direction for future work is the extension of our stability analysis to localized time-periodic solutions. Such structures have been observed numerically to exist and lie on isolas \cite{papangelo2017snaking}, and their existence has recently been rigorously established in spatially discrete complex-valued Ginzburg–Landau systems \cite{bergland2025localized}. While the methods developed here for steady-state solutions provide a natural starting point, analyzing the stability of time-periodic patterns would require Floquet theory, introducing additional technical challenges. Another limitation of the current work is that it does not directly handle fully infinite two-dimensional lattices; the multi-dimensional lattices we consider require boundedness in all but one spatial dimension. Extending the discrete Evans-function framework to address time-periodic solutions, as well as truly infinite lattices or more general lattice geometries, represents a promising avenue for future research.

\textbf{Competing Interests:} None

\bibliographystyle{abbrv}
\bibliography{references.bib}
\end{document}